\documentclass[reqno,11pt]{amsart}
\usepackage{mathtools}
\usepackage{microtype}
\usepackage[mathscr]{eucal}
\usepackage{slashed}
\usepackage{enumerate}
\usepackage{tikz}
\usepackage{amssymb}
\usepackage[colorlinks=true,linkcolor=black,citecolor=black]{hyperref}
\usepackage[sort&compress,capitalise,nameinlink]{cleveref}

\usepackage{ulem}
\usepackage{xcolor}
\usepackage{bm,bbm}

\usepackage{geometry}
\DeclareMathOperator{\curl}{curl}

\DeclareMathOperator{\Div}{div}

\DeclareMathOperator{\Real}{Re}

\renewcommand{\geq}{\geqslant}
\renewcommand{\leq}{\leqslant}

\newcommand{\C}{\mathbb C}
\newcommand{\Z}{\mathbb Z}
\newcommand{\R}{\mathbb R}
\newcommand{\N}{\mathbb N}

\newcommand{\cB}{\mathcal B}
\newcommand{\cE}{\mathcal E}
\newcommand{\cG}{\mathcal G}
\newcommand{\cH}{\mathcal H}
\newcommand{\cJ}{\mathcal J}
\newcommand{\cO}{\mathcal O}

\newcommand{\Ab}{\mathbf A}
\newcommand{\Fb}{\mathbf F}

\newcommand{\ii}{\mathrm i}
\newcommand{\ee}{\mathrm e}
\newcommand{\jb}{\mathbf j}

\def\XXint#1#2#3{{\setbox0=\hbox{$#1{#2#3}{\int}$ }
\vcenter{\hbox{$#2#3$ }}\kern-.6\wd0}}

\newtheorem{theorem}{Theorem}[section]

\newtheorem{lemma}[theorem]{Lemma}
\newtheorem{proposition}[theorem]{Proposition}
\newtheorem{definition}[theorem]{Definition}
\newtheorem{corollary}{Corollary}[section]
\theoremstyle{remark}
\newtheorem{remark}[theorem]{Remark}

\makeatletter
\@addtoreset{equation}{section}

\makeatother

\newcommand{\Hc}{H_{\mathrm{c}1}}
\newcommand{\Hcc}{H_{\mathrm{c}2}}
\newcommand{\Hccc}{H_{\mathrm{c}3}}

\numberwithin{equation}{section}
\newcommand{\bdm}{\begin{displaymath}}
\newcommand{\edm}{\end{displaymath}}
\newcommand{\bay}{\begin{array}{c}}
\newcommand{\eay}{\end{array}}
\newcommand{\ben}{\begin{enumerate}}
\newcommand{\een}{\end{enumerate}}
\newcommand{\beq}{\begin{equation}}
\newcommand{\eeq}{\end{equation}}
\newcommand{\beqn}{\begin{eqnarray}}
\newcommand{\eeqn}{\end{eqnarray}}
\newcommand{\bml}[1]{\begin{multline} #1 \end{multline}}
\newcommand{\bmln}[1]{\begin{multline*} #1 \end{multline*}}

\newcommand{\lf}{\left}
\newcommand{\ri}{\right}

\newcommand{\xv}{\mathbf{x}}

\newcommand{\av}{\mathbf{a}}

\newcommand{\eps}{\varepsilon}
\newcommand{\diff}{\mathrm{d}}

\newcommand{\hex}{h_{\mathrm{ex}}}

\newcommand{\nuv}{\bm{\nu}}

\definecolor{RoyalBlue}{RGB}{25, 25, 112} 
\definecolor{DarkBrown}{RGB}{101, 67, 33}
\newcommand{\hl}[1]{\textcolor{RoyalBlue}{#1}}

\title[Bulk Vortices in Strong Magnetic Fields]{Existence of Bulk Vortices in Superconductors with Strong Magnetic Fields}
\author[M. Correggi]{Michele Correggi}
\address[M. Correggi]{\newline Dipartimento di Matematica,  Politecnico di Milano,   P.zza Leonardo da Vinci,  32,  20133, Milan,  Italy.}
\email{michele.correggi@polimi.it}

\author[A.  Kachmar]{Ayman Kachmar}
\address[A. Kachmar]{\newline Department of Mathematics \normalfont{and} PDE Research Unit--Center for Advanced Mathematical Sciences (CAMS), American University of Beirut,
Beirut, Lebanon.}
\email{ak292@aub.edu.lb}
\date{\today}


\begin{document}

\begin{abstract}
We study the vortex formation in extreme type-II superconductors immersed in strong magnetic fields in the framework of the the Ginzburg-Landau theory. We focus on the regime where superconductivity survives in the bulk of the material but the magnetic field penetrates the sample, {i.e.}, for applied field much larger than the first critical one, but below the transition to surface superconductivity. Through a two-scale vortex construction, we obtain precise estimates for the vortex distribution and prove the existence of isolated defects with non-trivial winding numbers. In this respect, our work provides the first rigorous mathematical proof of the existence of isolated vortices for  fields comparable to the second critical one.
\end{abstract}

\maketitle

\section{Introduction}

\subsection{Superconductivity and Critical Magnetic Fields}

Since its discovery at the beginning of last century, superconductivity has become a widely studied phenomenon in condensed matter physics with many important practical applications (see, {e.g.}, the monograph \cite{Ti} for an extensive discussion). The macroscopic features of superconductivity, in particular close to the critical temperature of the phase transition, are extremely efficiently modeled in the framework of the phenomenological Ginzburg-Landau (GL) theory: the local conducting properties of the material are encoded in an {\it order parameter} $ \Psi $, {i.e.}, a complex-valued function whose modulus varies between 0 and 1. The former value identifies the normal state, where the material behaves like a normal conductor, while if $ |\Psi| = 1 $ at some point there is perfect superconductivity. At equilibrium, the order parameter is obtained as the minimizer of a free energy functional -- the GL functional, see next \cref{sec: GLf}. We note that GL theory can be derived rigorously from the microscopic BCS theory \cite{DHM1,DHM2,3} for temperatures close to the superconducting critical temperature (see also \cite{CC} and references therein for a similar but different regime).

In general, the order parameter inside a sample is non-constant, and isolating the regions of superconductivity loss (if any) is thus quite important. This becomes even more relevant in the presence of {\it external magnetic fields}, which are capable of destroying the superconductive behavior of the material. The physical phenomenology of superconductors in magnetic fields is the following (see, {e.g.}, \cite{DeG}):
\begin{itemize}
	\item small magnetic fields do not affect the superconducting behavior, and the field is repelled by the sample ({\it Meissner effect});
	\item intermediate fields penetrate the sample at isolated defects, where superconductivity is lost;
	\item strong fields eventually destroy superconductivity.
\end{itemize}

This phenomenology can be entirely reproduced within the GL theory, at least for extreme type-II superconductors. The GL energy functional must  be modified by adding another physical parameter -- the {\it induced vector potential} $ \Ab $ --, which provides the magnetic field $ \curl \Ab $ inside the sample. The equilibrium configuration is thus a suitable pair $ (\Psi_*, \Ab_*) $ minimizing the modified GL functional. The transitions mentioned above can actually be proven and one can associate to each one a {\it critical magnetic field} marking the threshold for a macroscopic change of the GL state. 
The {\it first critical field} $ \Hc $ is determined by the appearance of defects in the order parameter; the {\it second critical field} $ \Hcc $ identifies the transition from bulk to boundary behavior; the {\it third critical field} $ \Hccc $ marks the disappearance of superconductivity. We refer to \cite{FH,SS07} for a mathematically rigorous definition of such fields and many other results concerning the GL theory.

In this paper we focus on the so-called {\it bulk regime}, where isolated defects occur inside the sample of the material but superconductivity is preserved everywhere else. Similarly, the magnetic field penetrates the sample only in correspondence with such defects and vanishes elsewhere. Concretely, this is expected to happen whenever the magnetic field is above $ \Hc $ and much larger than it, but below $ \Hcc $: Here we consider magnetic fields of the same order of $ \Hcc $ but strictly below it. In this regime, each defect is believed to coincide with a {\it vortex} of the order parameter, {i.e.}, a zero of $ \Psi $ with nontrivial winding number around it. Such a result is however unproven for such strong fields (see \cref{sec: GLf} below for a more detailed discussion of the literature), and only the distribution of the energy proved to be accessible so far. Our main contribution is precisely a proof of the existence of small balls encircling one or more vortices which are approximately uniformly distributed in the sample and almost cover the region where the modulus of the order parameter is significantly smaller than 1.

\subsection{The Ginzburg-Landau functional}
\label{sec: GLf}

We study a virtually infinite superconducting wire which is straight along the $ z-$axis and has constant cross section $ \Omega $. The external magnetic field is assumed to be parallel to the wire axis.

Let then  $\Omega\subset\R^2$ be a bounded simply connected domain with a smooth boundary $\partial\Omega$, and let $\nuv$ be the unit normal vector on $\partial\Omega$ pointing inward $\Omega$. We introduce the following minimization domain\footnote{The vanishing of the divergence and the boundary condition can be imposed without loss of generality: starting with the larger space $H^1(\Omega;\mathbb{C}) \times H^1(\Omega;\mathbb{R}^2)$ yields the same minimizers up to a gauge transformation \cite[p. 143, Eq. (10.7)]{FH}.}
\begin{equation}\label{eq:def-space}
\cH=\lf\{ \lf(\Psi,\Ab \ri)\in H^1(\Omega;\C)\times H^1(\Omega;\R^2)\colon \Div\Ab=0\mbox{ in }\Omega, ~\nuv\cdot\Ab=0\mbox{ on }\partial\Omega \ri\},
\end{equation}
and we denote by $\Fb$ the unique vector field satisfying \cite[Propositions~D.2.1 and D.2.2]{FH}
\begin{equation}\label{eq:def-F}
\curl\Fb=1,\quad \Div\Fb=0\mbox{ in }\Omega,\quad \nuv\cdot\Fb=0\mbox{ on }\partial\Omega.
\end{equation}

Let us fix $b\in(0,1)$. In appropriate units where the GL parameter reads $\kappa=1/\eps$ and the magnetic field  intensity is $h_{\mathrm{ex}} =b\varepsilon^{-2}$, with $0 < \eps \ll 1 $, the GL energy functional on $\Omega$ is defined on the space $\cH$ as
\begin{equation}\label{eq:GL-f}
\cE_{b,\eps}(\Psi,\Ab)=\int_\Omega \lf( \lf| \lf(\nabla-\ii b \eps^{- 2} \Ab  \ri)\Psi \ri|^2+\frac1{2\eps^2}(1-|\Psi|^2)^2+\frac{b^2}{\eps^{4}}|\curl(\Ab-\Fb)|^2 \ri)\diff \xv.
\end{equation}
We are going to denote by $ \cE_{b,\eps}(\Psi,\Ab;\mathcal{D}) $ the restriction of the above energy to a domain $ \mathcal{D} \subset \Omega $.

As anticipated, the physical interpretation of $(\Psi,\Ab)$ is that $|\Psi|^2$ and $\curl\Ab$ measure the local density of superconductivity and the induced magnetic field, respectively, and the supercurrent is given by
\[\jb(\Psi,\Ab)=\Real \lf\langle \ii\Psi,(\nabla-\ii b\eps^{-2}\Ab)\Psi \ri\rangle_\C, \]
where $\langle z,z'\rangle_\C= z\overline{z'}$ for two complex numbers $z$ and $z'$. If we write $\rho=|\Psi|$ and  $\Psi=\rho\ee^{\ii\varphi}$, the supercurrent reads as
\[\jb(\Psi,\Ab)=\rho^2(\nabla\varphi-b\eps^{-2}\Ab),\]
making apparent its proportionality to the gradient of the phase of the order parameter and to the magnetic potential.

A minimizing configuration $ (\Psi_*,\Ab_*) \in \cH $ is such that 
\bdm
	\cE_{b,\eps}(\Psi_*,\Ab_*) = \inf_{(\Psi,\Ab)\in \cH}\cE_{b,\eps}(\Psi,\Ab).
\edm
If $(\Psi_*,\Ab_*)$ is such a minimizer of $\cE_{b,\eps}$, then it is well known that $|\Psi_*|\leq 1$ and the following system of equations holds
\begin{equation}\label{eq:GL}
\begin{cases}
-(\nabla-\ii b\eps^{-2}\Ab_*)^2\Psi_*=  \eps^{-2}(1-|\Psi_*|^2)\Psi_*,		&	\mbox{ in }\Omega,\\
-\nabla^\bot\curl\Ab_*=b^{-1}\eps^{2}\jb(\Psi_*,\Ab_*),					&	\mbox{ in }\Omega,\\
\nuv\cdot\nabla\Psi_*=0,													&	\mbox{ on }\partial\Omega,	\\
\curl\Ab_*=1, 															&	\mbox{ on }\partial\Omega,
\end{cases}
\end{equation}
where $\nabla^\bot=(-\partial_2,\partial_1)$. To stress the dependence on $b$ and $\varepsilon$, we sometimes use the notation $(\Psi_*,\Ab_*)_{b,\varepsilon}$ for a minimizer of $\cE_{b,\varepsilon}$.

In the units we have chosen, the asymptotics of the critical magnetic fields mentioned in the previous Section are know exactly \cite{SS07,FH}. 
Specifically, the critical values of \(h_{\mathrm{ex}}\) are
\[
H_{c1} \sim C_{\Omega}|\log\eps|,\qquad H_{c2}\sim \eps^{-2},\qquad H_{c3}\sim \eps^{-2}/0.59.
\]
In the setting described here, the first critical field $ \Hc$ and its precise asymptotic expression are discussed in detail in \cite{S99, SS00, SS03} and \cite[Chpt. 7]{SS07} (see also \cite{BJOS, RSS} for the three-dimensional analogue of the problem): for applied fields below $ \Hc $ the order parameter does not vanish, while it has isolated zeros -- {\it vortices} -- above $H_{c1}$. If the applied field's intensity stays between $H_{c2}$ and $H_{c3}$, the order parameter concentrates on the boundary of $\Omega$ and the sample exhibits {\it surface supercondutivity}. This regime was first addressed in \cite{Pan} and then extensively studied later, {e.g.}, in \cite[Chpt. 14]{FH}, \cite{AH, CR1, CR2} and references therein (see also  the review \cite{Cor} for a discussion of domains with corners at the boundary). {The transition occurring close to $H_{c2}$, as well as the distribution of superconductivity, was investigated in \cite{Al, AS, FK2, K-sima}.}  Eventually,  above $H_{c3}$, superconductivity is entirely lost and the minimizer is the {\it normal state} $(0,\Fb)$ (see, {e.g.}, \cite[Chpt. 13]{FH} and references therein).

In this paper, the applied magnetic field is 
\bdm
	\hex= \frac{b}{\eps^{2}},	\qquad		0 < b < 1,
\edm
so that we are in the regime below but with a field of the same order as $H_{c2}$. Let us describe the expected behavior of GL minimizers in this regime: as long as $ \hex \sim |\log\eps| $ vortices do not cover the whole domain but they are confined in a region which can be determined by solving a suitable free-boundary problem \cite[Chpt. 7]{SS07}; however, as soon as $ \hex \gg |\log\eps| $, vortices are expected to  occupy uniformly the whole domain $ \Omega $ with a constant density proportional to $ \hex $. They are also believed to arrange themselves in a periodic triangular lattice -- the famous {\it Abrikosov lattice} \cite{Ab} --, whose existence is observed experimentally but mathematically only conjectured {as a minimizing configuration (see also the related mathematical works \cite{ST1, ST2} on the existence and stability of Abrikosov lattice solutions to the GL equations).}

Compared to the physical heuristics, the mathematical literature on this regime is much more scarce. The first work to mention is \cite{SS02}, where it is shown that the energy is uniformly distributed and it is given by an effective model on a cell (see next \cref{sec: effective}). We stress however that such a result holds on a scale $ \gg \eps $, which is the expected sized of vortex cores but it is also proportional to the mean distance between vortices $ \propto \frac{1}{\sqrt{\hex}} $. In fact, this is the very reason why the regime we are considering is so difficult to address, because the expected vortex lattice has a spacing of the same order as the vortex core. Indeed, if $ |\log\eps|\ll \hex \ll \eps^{-2} $, much more is known, and it was even proven that the vorticity of the order parameter has uniform distribution $ \hex $ \cite{SS001}.

Two key contributions about the regime under consideration are also the works \cite{FH,HK}, where several fundamental properties of any minimizing configuration are proven. Despite the progress, however, the features of the vortex structure of the order parameter as well as the behavior of the magnetic field inside the sample have been elusive so far. In this paper, we address precisely this question and partially fill the gap by proving that if $ b < b_0 $, with $ b_0 $ independent of $ \eps $, the order parameter has vortices in the bulk of the sample.

\subsection{The effective energy}
\label{sec: effective}

Let us recall an effective energy introduced in \cite{SS02} (see also \cite{AS, FK}). Consider the square\footnote{The original effective model was defined in a ball of large radius, but the shape of the boundary is expected to play no role, provided it grows isotropically (see \cite{AS, FK}).} $K_R=(-R/2,R/2)^2$ of side length $R$. For $u\in H^1(K_R;\C)$, we set 
\begin{equation}\label{eq:def-G}
\begin{aligned}
    \cG_{b,R}(u)&=\int_{K_R} \lf(b|(\nabla-\ii\Ab_0)u|^2-|u|^2+\tfrac12|u|^4 \ri)\diff\xv\\
    &=\int_{K_R}\lf(b|(\nabla-\ii\Ab_0)u|^2+\tfrac12(1-|u|^2)^2-\tfrac12 \ri)\diff\xv,
\end{aligned} 
\end{equation}
where $\Ab_0(\xv)=\frac12(-x_2,x_1)$. Then, the  limit
\[g(b)=\lim_{R\to+\infty}\left(\frac1{|K_R|}\inf_{u\in H^1_0(K_R,\C)}\cG_{b,R}(u)\right)\]
exists and is finite. Moreover, the function  $g$ is concave and increasing \cite[Thm.~2.1]{FK},  and  as $b\to0_+$, it satisfies\footnote{To avoid confusion, we use the following convention: a quantity $ \alpha = o_{\beta}(\gamma) $, for $ \beta, \gamma \in \R^+ $, if $ \limsup_{\beta \to 0_+} \gamma^{-1} |\alpha| = 0 $. An analogous convention will be used for $ \mathcal{O} $. If there are other parameters involved the estimate is supposed to hold for all the values of the parameters in the specified range, but not necessarily uniformly, unless explicitly stated otherwise.} \cite[Thm.~2.4]{K-jfa}\footnote{In \cite{SS02, K-jfa}, the authors work with $f(b)=\frac12\bigl(g(b)+\frac12\bigr)$.}
\begin{equation}\label{eq:asym-g}
g(b)=-\tfrac12-\tfrac{b}2\log b +o_{b}(b\log b).
\end{equation}
As noted in \cite{SS02}, the term \(\tfrac{b}2 \log b\) in \eqref{eq:asym-g} is a logarithmic correction to \(g(0)=-\tfrac{1}{2}\) and signals the presence of vortices. As \(b\to1^-\), the behavior of \(g(b)\) involves the Abrikosov energy, further confirming the vortex structure. This subtle asymptotic behavior was first identified in \cite{AS} and later revisited in \cite{FK}. A parallel phenomenon appears in the context of fast rotating Bose–Einstein condensates in the lowest Landau level regime \cite[Theorem~3.9]{4}, where a correction term in the energy asymptotics likewise marks the emergence of the Abrikosov lattice of vortices.

The concavity of $g$ ensures that $g$ is continuous and that the left- and right-derivatives, $g'(b_-)$ and $g'(b_+)$, exist on $(0,1)$; moreover, the set
\begin{equation}\label{def:J}
\cJ=\{b\in(0,1)\colon g'(b_-)\not=g'(b_+)\}
\end{equation}
is at most countable. Therefore, we know that $g$ is differentiable at every $b\in(0,1)\setminus \cJ$, and we may consider the function $g':(0,1)\setminus \cJ\to\R$.

The relevance of $g(b)$ is that, for $0<b<1$, a minimizer $(\Psi_*,\Ab_*)$ of $\cE_{b,\eps}$ satisfies as $\eps\to0$,
\begin{equation}\label{eq:bulk}
\begin{gathered}
\cE_{b,\eps}(\Psi_*,\Ab_*)=\lf[ |\Omega|\bigl(g(b)+1/2 \bigr)+o_{\eps}(1) \ri]\eps^{-2} ,\\
b^2\int_\Omega|\curl(\Ab_*-\Fb)|^2\diff\xv=o_{\eps}(\varepsilon^2),\\
\int_\Omega|\Psi_*|^4\diff \xv=-2g(b)|\Omega|+o_{\eps}(1),
\end{gathered}
\end{equation}
and if $b\in(0,1)\setminus\cJ$,
\begin{equation}\label{eq:bulk*}
\int_\Omega |\Psi_*|^2\diff\xv=\bigl(bg'(b)-2g(b)\bigr)|\Omega|+o_{\eps}(1).
\end{equation}
The  asymptotics in \eqref{eq:bulk} were proven in \cite{SS02}, while the one in \eqref{eq:bulk*}  was proven in \cite{HK}.\medskip

In this paper, we prove the following new formula for the derivative of $g$, consistent with a formal differentiation of \eqref{eq:asym-g}.
\begin{theorem}[Derivative of $ g $] 
	\label{thm:g(b)}
	\mbox{}	\\
	As $b\to0_+$, we have
	\beq
		g'(b)=-\tfrac12\log b \lf(1 + o_b(1)\ri).
	\eeq
\end{theorem}

As a direct consequence of \cref{thm:g(b)} and \eqref{eq:bulk}-\eqref{eq:bulk*}, we get the following

\begin{corollary}
	\label{corol:potential}
	There exists $ b_0 > 0 $ and a function  $r:(0,b_0]\to\R_+$, such that $\lim_{b\to0_+}r(b)=0$, and, for $b<b_0$ and any minimizing configuration  $(\Psi_*,\Ab_*)_{b,\varepsilon}$,
	\beq
		\limsup_{\varepsilon\to0_+}\int_\Omega \lf(1-|\Psi_*|^2 \ri)^2\diff \xv \leq b|\log b|r(b).  
	\eeq
\end{corollary}

Another corollary is that we rule out a limit constant profile of $|\Psi_*|$.
\begin{corollary}
    \label{corol:density}
    There exists an infinite subset $\mathcal C\subset(0,1)$ such that, if $ b\in\mathcal C$ and  $(\Psi_*,\Ab_*)_{b,\varepsilon}$ is a minimizing configuration, then, $|\Psi_*|$ does not converge pointwise as $\varepsilon\to0$ to a constant function, even along a subsequence.
\end{corollary}
The lack of homogeneity in the profile of \(|\Psi_*|\) indicates that the vortex cores, whose typical size is of order 
\(\varepsilon\), cannot be much smaller than their separation. Indeed, if the cores were negligible compared to their spacing, the modulus would tend to a constant away from the cores, contradicting \cref{corol:density}. This absence of scale separation is the main obstacle in detecting vortices in the regime we study.

\subsection{Vortex structure}
For a given minimizing configuration \((\Psi_*,\Ab_*)\) of the functional
\(\cE_{b,\varepsilon}\), a {\it vortex} is an isolated zero of \(\Psi_*\), 
which carries a nontrivial winding number around.

In order to isolate such topological defects, the typical mathematical method relies on the identification of {\it vortex balls}, which we define rigorously below. For the first application of such a strategy we refer to \cite{BBH}, while the key contributions for the techniques we use here are the works \cite{J,Sa} (see also \cite[Chpt. 4]{SS07}), where the vortex ball construction was perfected. 
\begin{definition}[Vortex balls] \label{def:vortex-ball}  
\mbox{}	\\
Let \((\Psi_*,\Ab_*)_{b,\varepsilon}\) be a minimizing configuration of \eqref{eq:GL-f}. An open ball \(B(\av,\varrho)\) is called a vortex ball in \(\Omega\) with degree \(d \in \Z \), if  
\[  
\overline{B(\av,\varrho)} \subset \Omega, \qquad |\Psi_*| \geq \tfrac{1}{2} \text{ on } \partial B(\av,\varrho), \qquad \mathrm{deg}\bigl(\Psi_*, \partial B(\av,\varrho)\bigr) = d.  
\]  
\end{definition}  
For any vortex ball \(B(\av,\varrho)\) with non-zero degree, \(\Psi_*\) must vanish at least once inside \(B(\av,\varrho)\). Informally, we identify the center \(\av\) of the ball as the vortex {\it  center} and its radius \(\varrho\) as the vortex {\it core}.  

Understanding the vortex ball structure --particularly their distribution-- is central to the analysis of the GL model. While extensively studied, the problem remains open for fixed \(b > 0\). Using the asymptotic formula \eqref{eq:asym-g} and a refined version of the bulk estimate in \eqref{eq:bulk}, we develop a two-scale vortex construction, which consists of two main steps:
\begin{enumerate}
    \item ({\it localization}) we restrict the analysis to a square $Q_\ell$ of side length \(\ell\) with \(\varepsilon \ll \ell \ll 1\);  
\item ({\it rescaling}) zooming into the square \(Q_\ell\), we construct vortices with core size \(b^{-1/2}|\log b|^{-2} \varepsilon\) and bounded degree, for small but fixed \(b\).  
\end{enumerate}
The second step adapts the framework of \cite[Chpt. 4 and 6]{SS07}, treating their vortex construction as a black box.

To state our result, we recall the following standard notion:
\begin{enumerate}[--]
\item $ \mathscr{L} $ denotes the Lebesgue measure on $\Omega$;
\item $\delta_{\mathbf a}$ denotes the Dirac measure centered at $\mathbf a$;
\item $C_0^{0,1}(\Omega)$ denotes the space of compactly supported Lipschitz continuous functions on $\Omega$, endowed with the Lipschitz norm
\beq
	\label{eq: Lipschitz}
	\displaystyle\|f\|_{\mathrm{Lip}}=\sup_{\substack{\mathbf x,\mathbf y\in\Omega\\ \mathbf x\not=\mathbf y}}\frac{|f(\mathbf x)-f(\mathbf y)|}{|\mathbf x-\mathbf y|};
\eeq
\item ${C_0^{0,1}(\Omega)}^*$ denotes the topological dual of $C_0^{0,1}(\Omega)$, with the standard norm.
\end{enumerate}

%
\begin{theorem}[Vortex ball construction]\label{thm:main}
\mbox{}	\\
There exists a function $\mathrm r:(0,b_0]\to\R_+$, such that $\lim_{b\to0_+}\mathrm r(b)=0$, and, for any given $b_1\in(0,b_0]$, there exists a constant $\varepsilon_0>0$,  such that  the following holds:
\begin{enumerate}[\rm i)]
\item 
if  $(\Psi_*,\Ab_*)_{b,\varepsilon}$ is a minimizing configuration of \eqref{eq:GL-f} for $b_1<b<b_0$ and $0<\varepsilon<\varepsilon_0$, then there exist  disjoint vortex balls $\lf(B(\mathbf a_i,\varrho_i)\ri)_{i\in\mathcal I}$ in $\Omega$ with positive degrees $d_i \in \N_0 $ and radii satisfying $|\varrho_i|\leq b^{-\frac12}|\log b|^{-2} \varepsilon$, for any $ i\in\mathcal I $;
\item as $\varepsilon\to0_+$, we have
\beq
	\label{eq: vortex distribution}
	\lf\|b^{-1}\varepsilon^2 \sum_{i \in \mathcal{I}} 2\pi d_i\delta_{\mathbf a_i}- \mathscr{L} \ri\|_{C_0^{0,1}(\Omega)^*} \leq \mathrm{r}(b)+o_{\eps}(1),
\eeq
uniformly with respect to $b\in[b_1,b_0]$.
\end{enumerate}
\end{theorem}

\begin{remark}[Occurrence of vortices]
	\mbox{}	\\
	As already anticipated, the above result ensures that any minimizer $ \Psi_* $ contains vortices. In fact, inside each vortex ball provided by \cref{thm:main} there may be more than one vortex, but the winding numbers of all the vortices sums up to $ d_i \geq 1 $. In addition, their weighted distribution in $ \Omega $ is asymptotically uniform for $ b $ small. One actually expects that the winding number of each vortex is exactly 1 and therefore the weight $ d_i $ in \eqref{eq: vortex distribution} should be irrelevant. We also remark that not all the vortices are covered by the family of balls $ \lf(B(\mathbf a_i,\varrho_i)\ri)_{i\in\mathcal I}$, because in the localization step we have to discard a number of bad cells where we cannot perform the vortex ball construction. Additionally, we remark that the bound on the vortex ball radii is optimal with respect to $\varepsilon$ but likely not optimal with respect to $b$. Indeed, some heuristic estimate of a single vortex energy seems to suggest that the area covered by vortex ball should vanish as $ b \to 0_+ $.
\end{remark}

\begin{remark}[Vorticity measure \& supercurrent]
	\mbox{}	\\
	In the proof of the above result, we actually prove an estimate about the {\it vorticity measure} associated to $ \Psi_* $, i.e., the quantity
	\beq
		\label{eq: vorticity}
		\mu : = \curl \jb(\Psi_*,\Ab_*) + b\varepsilon^{-2}\curl \Ab_*
	\eeq
	which, if $ \Psi_* $ was a map to $ \mathbb S^1 $, would be given by the sum of Dirac measures centered at the vortices weighted with the corresponding winding numbers, i.e., what appears inside \eqref{eq: vortex distribution}, up to the rescaling factor $ b^{-1}\varepsilon^2 $. As discussed in \cite[Chpt. 6]{SS07}, the two quantities are in fact close as $ \eps \to  0 $, and this is why the estimate \eqref{eq: vortex distribution} gives also information about the supercurrent associated to $ \Psi_* $. In particular, a byproduct of the proof of \cref{thm:main} obtained by integrating $ \mu $ over $ \Omega $ is the following:
	\bml{
		\label{eq: boundary current}
		\frac{1}{|\Omega|} \int_{\partial \Omega} \bm{\tau} \cdot \lf( \jb(\Psi_*, \Ab_*) + {b\varepsilon^{-2}}\Ab_* \ri) = \frac{1}{|\Omega|} \int_{\partial \Omega} \bm{\tau} \cdot {\Real \lf\langle \ii\Psi_*, \nabla \Psi_* \ri\rangle_\C} \\
		= {\frac{b}{\eps^2} \lf(1+\cO_b( \mathrm{r}(b)) + o_{\eps}(1) \ri)},
	}
	i.e., there is a huge current circulating at the boundary of the sample, compensating the magnetic field inside. More detailed asymptotics for the current circulation are available in two complementary regimes: 
the surface regime \(b > 1\) and the regime far below \(H_{c2}\) where \(b(\varepsilon) \ll 1\). 
For instance, in the former case it is known (see \cite[Theorem~3]{CR2}) that the order parameter carries a huge phase at the boundary whose leading term is proportional to $ |\Omega| \eps^{-2} $, or, equivalently the boundary current is $ 1/\eps^2 $ to leading order. This is in perfect agreement with \eqref{eq: boundary current} as $ b \to 1_- $.
\end{remark}

\begin{remark}[Higher values of $b$]
\mbox{}	\\
{Thanks to \eqref{eq:GL}, we can express the vorticity measure in terms of the induced magnetic field $h_*=\curl\Ab_*$ via
\[b^{-1}\varepsilon^2\mu=-\Delta h_*+h_*,\]
with $h_*=1$ on $\partial\Omega$. By \cite[Proposition~11.4.4]{FH}, $h_*$ is bounded in $C^2(\overline{\Omega})$ and converges to $1$ in $C^1(\overline{\Omega})$, thus
\[b^{-1}\varepsilon^2\mu \xrightarrow[\eps \to 0]{} \mathscr{L}\]
in the sense of distributions. Moreover, integrating $\mu$ over $\Omega$, we obtain 
\[\frac{1}{|\Omega|} \int_{\partial \Omega} \bm{\tau} \cdot \lf( \jb(\Psi_*, \Ab_*) + b\varepsilon^{-2}\Ab_* \ri) = \frac{b}{\eps^2} \lf( 1 + o_{\eps}\lf( 1\ri) \ri).\]
While this holds for all $b\in(0,1)$, the proof of \cref{thm:main} establishes a relation between the vorticity measure and vortex balls specifically for $b\in(0,b_0)$.}
\end{remark}

\begin{remark}[Error function $ \mathrm{r}(b) $]
	\mbox{}	\\
	The expression of the function $\mathrm r(b)$ is explicitly given as
	\[ \mathrm r(b)=Cr_0(b),\]
	where $C$ is a positive constant and
	\begin{equation}\label{eq:def-r0}
		r_0(b)=\max\left\{|\log b|^{-1/2},\biggl|\frac{g(b)+\frac12}{b|\log b|}-\frac12\biggr|,\frac{\log|\log b|}{|\log b|}\right\}.
	\end{equation}
\end{remark}

\section{Proof of Theorem~\ref{thm:g(b)}}

\subsection{Groud states with magnetic periodic conditions}

We need to work with a periodic minimizer of the functional $\cG_{b,R}$. Due to the presence of the magnetic term, we impose magnetic periodic conditions (see \cite{Du, Od} and more recently \cite[Section~II]{5}). The underlying idea is that, since the magnetic Laplacian $-(\nabla - i\mathbf{A}_0)^2$ does not commute with ordinary translations, one restricts to an eigenspace of magnetic translations. This leads to the definition of the space
\begin{multline}\label{eq:space-p}
\mathcal{H}_R = \lf\{ u\in H^1_{\rm loc}(\R^2;\C) \colon u(x_1+R,x_2)=\ee^{\ii Rx_2/2}u(x_1,x_2), \ri.\\
\lf. u(x_1,x_2+R)=\ee^{-\ii Rx_1/2}u(x_1,x_2) \mbox{ a.e. on }\R^2 \ri\}.
\end{multline}
Heuristically, the idea is that the fine structure of any GL minimizing configuration $ (\Psi_*, \Ab_*) $ is periodic on the cell scale $ \ell $, and therefore imposing periodic boundary conditions does not affect the behavior of the effectiveness of the model problem. 
Note that  $\cG_{b,R}(u)$ is well defined on $ \mathcal{H}_R $ and, if $u\in \mathcal{H}_R$, the physically relevant quantities
\[|u|, \qquad  \mathbf j(u)= \mathrm{Re} \lf\langle\ii u, \nabla u \ri\rangle_\C , \qquad \curl \Ab_* \]
are periodic over the square lattice generated by $K_R$. 

Minimizing $\cG_{b,R}$ over $\mathcal{H}_R$, we get  the existence of a minimizer by the direct method of the calculus of variations (see \cite{Od}).
Moreover, it is  know \cite{AS, FK} that 
\begin{equation}\label{eq:g-periodic} 
g(b)=\lim_{R\to+\infty}\left(\frac1{|K_R|}\inf_{u\in \mathcal{H}_R}\cG_{b,R}(u) \right).\end{equation}
In the sequel, we denote by $u_*\in \mathcal{H}_R$ a minimizer of $\cG_{b,R}$, i.e. 
\[\cG_{b,R}(u_*)=\inf_{u\in \mathcal{H}_R}\cG_{b,R}(u).\]
Writing
\begin{equation}\label{eq:asy-p}\begin{gathered}
\int_{K_R}|(\nabla-\ii \Ab_0)u_*|^2\diff \xv=:\bigr(g'(b)+f_1(R,b)\bigl)|K_R|,\\
\int_{K_R}|u_*|^2\diff \xv=:\bigr(bg'(b)-2g(b)+f_2(R,b)\bigl)|K_R|,\\
\int_{K_R}|u_*|^4\diff \xv=:\bigr(-2g(b)+f_3(R,b)\bigl)|K_R|,\\
\end{gathered}
\end{equation}
we then have by \cite[Prop.~3.1 \& Thm.~4.4]{HK} that for every $b\in(0,1)\setminus\cJ$,
\[f_i(R,b)\xrightarrow[R\to+\infty]{} 0 \qquad(i=1,2,3).\]
\subsection{Optimal upper bounds}

 It follows  from \eqref{eq:asy-p}  that
\begin{equation}\label{eq:int-g-g'}
0\leq \frac12\int_{K_R}(1-|u_*|^2)^2\diff\xv=\Bigl(\tfrac12-bg'(b)+g(b) \Bigr)|K_R|+o_R(R^2),
\end{equation}
and thanks to  \eqref{eq:asym-g},  we deduce the upper bound $g'(b)\leq -\frac12\log b+o_{b}(\log b)$. In \cref{prop:g-ub} below, we give a significant improvement of the remainder estimate.

\begin{proposition}\label{prop:g-ub}
As $b\to0_+$, we have 
\beq
	g(b)\leq -\tfrac12-\tfrac{b}{2}\log b+\cO(b)\qquad\mbox{and}\qquad g'(b)\leq -\tfrac12\log b+\cO(1).
\eeq    
\end{proposition}

\begin{remark}[Upper bound on $ g(b) $]\label{rem:g-ub}
\mbox{}	\\
The upper bound on $g(b)$ was proven already in \cite[Thm. 1.4]{SS02} but we revisit the construction of the trial state to highlight the role of the magnetic periodic boundary conditions.
We construct a trial state in a variant of the space $\mathcal{H}_R$, which exploits further the gauge invariance of the functional. Specifically, for given $\alpha,\beta\in\R$, we introduce
\begin{multline}\label{eq:space-p-alpha}
\mathcal{H}_R^{\alpha,\beta}= \lf\{v\in H^1_{\rm loc}(\R^2,\C)\colon v(x_1+R,x_2)=\ee^{\ii\alpha}\ee^{\ii Rx_2/2}v(x_1,x_2),	\ri. \\
\lf. v(x_1,x_2+R)=\ee^{\ii\beta}\ee^{-\ii Rx_1/2}v(x_1,x_2) \mbox{ a.e. on }\R^2 \ri\}.
\end{multline}
This space corresponds to a joint eigenspace of the magnetic translation generators, with eigenvalues $e^{\ii\alpha}$, $e^{\ii\beta}$ respectively. We note that $\mathcal{H}_R=\mathcal{H}_R^{0,0}$. If $v\in \mathcal{H}_R^{\alpha,\beta}$, then
\beq
	\label{eq: u and v}
	u=\ee^{-\frac{\ii \alpha x_1}{R}}\ee^{-\frac{\ii \beta x_2}{R}}v\in \mathcal{H}_R 
\eeq
and
\[ \mathcal G_{b,R}(u)=\mathcal G_{b,R}^{\alpha,\beta}(v):=\int_{K_R} \lf(b|(\nabla-\ii\Ab_0-\ii R^{-1}\mathbf k_{\alpha,\beta}) v|^2-|v|^2+\tfrac12|v|^4 \ri)\diff\xv,\]
where $\mathbf k_{\alpha,\beta}=(\alpha,\beta)$ is a constant vector.
\end{remark}

\begin{proof}[Proof of \cref{prop:g-ub}]
Using the definition of $\cG_{b,R}$ in \eqref{eq:def-G} and the first formula in  \eqref{eq:asy-p}, we write (recall \eqref{eq:def-G})
\begin{equation}\label{eq:proof-g-ub}
\cG_{b,R}(u_*)\geq \Bigl(bg'(b)-\tfrac12+f_1(R,b)\Bigr)|K_R|.
\end{equation}
Next, we construct a test configuration $u\in \mathcal{H}_R$ so that 
\begin{equation}\label{eq:proof-g-ub*}
\cG_{b,R}(u)\leq \Bigl(\tfrac{b}{2}|\log b|-\tfrac12+\cO(b)\Bigr)|K_R|+o(|K_R|).
\end{equation}
Since $\cG_{b,R}(u_*)\leq \cG_{b,R}(u)$, we get $g(b)\leq -\tfrac12+\tfrac{b}{2}|\log b|+\cO(b)$, and it results from \eqref{eq:proof-g-ub} and \eqref{eq:proof-g-ub*} that
\[ bg'(b)\leq \tfrac{b}2|\log b|+\cO(b),\]
which completes the proof of the proposition. Moreover, we obtain an improvement of \eqref{eq:int-g-g'} as follows,
\begin{equation}\label{eq:int-g-g'-*}
\frac12\int_{K_R}(1-|u_*|^2)^2\diff\xv\leq \tfrac{b}2|\log b|+b\cO_{R}(|K_R|)+o_{R}(|K_R|).
\end{equation}

It remains to construct $u$ and prove the upper bound \eqref{eq:proof-g-ub*} to conclude the argument. The construction of the trial state $u$ is inspired by \cite[Sect.~3]{ASa}. Let us suppose that $R\in\sqrt{2\pi}\N$ and set $N=R/\sqrt{2\pi}$. This would prove the inequality on a subsequence $ \lf\{ R_N \ri\}_{N \in \N} $, but the existence of the limit implies that one can restrict to a subsequence.
We decompose the square $K_R$ into $N^2$ congruent  squares $(Q_j)_{1\leq j\leq N^2}$ with   pairwise disjoint interiors. We also denote by $(Q_j)_{j\in\Z}$  the covering of $\R^2$ generated by the square $Q_1$, and for each $j$, we denote by $\av_j$ the center of the square $Q_j$, whose area is $2\pi$. 

Let $h$ be the solution to 
\[-\Delta h= 2\pi\delta_{\av_1}-1 \quad \mbox{ on } \quad Q_1,\]
with periodic boundary conditions. Note that this is possible because $\langle 2\pi\delta_{\av_1}-1,1\rangle_{\mathcal{S}'}=0$. We extend $h$ by periodicity to $\R^2$ and notice that it satisfies
\[\begin{cases}
    -\Delta h=2\pi\sum_{j\in\Z}\delta_{\av_j}-1\\
    h(x_1+\sqrt{2\pi},x_2)=h(x_1,x_2),	\mbox{ for a.e. } x_2\\
     h(x_1,x_2+\sqrt{2\pi})=h(x_1,x_2),	\mbox{ for a.e. } x_2
\end{cases}.\]
In particular, $h$ satisfies periodic boundary conditions on the boundary of $K_R$, and 
\[-\Delta h=2\pi\sum_{1\leq j\leq N^2}\delta_{\av_j}-1\qquad\mbox{on }K_R.\]
Hence, the function  $w(\xv)=h(\xv)-\log|\xv-\av_1|$ satisfies $-\Delta w+1=0$ on $Q_1$, and it is smooth on $Q_1$. 
Knowing that $h(\xv)=\log|\xv-\av_1|+w(\xv)$ with $w$ a smooth function, we can then verify  that,  as $b\to0_+$,
\begin{equation}\label{eq:est-h-log-b}
\begin{gathered}
\int_{Q_1}|\nabla h|\diff \xv\leq  \cO_{b}(1),\\
\int_{Q_1\setminus B(\av_1,\sqrt{b})}|\nabla h|^2\diff \xv\leq 2\pi|\log\sqrt{b}|+\cO_{b}(1),\\
\int_{ B(\av_1,\sqrt{b})}|x-\av_1|^2|\nabla h|^2\diff \xv\leq \cO_{b}(b),
\end{gathered}
\end{equation}
and, up to a conservative field, $-\nabla^\bot h+\Ab_0$ is given on $Q_1$ by 
\[\Fb_{\av_1}^{\rm AB}(\xv)=\frac{(\xv-\av_1)^\bot}{|\xv-\av_1|^2}\]
i.e., an Aharonov-Bohm vector potential with pole $\av_1$ and flux $1$.

Since the vector field $-\nabla^\bot h+\Ab_0$ has zero $\curl$ on  $\R^2\setminus\cup_j\{\av_j\}$, there is a multi-valued function $\phi$ defined modulo $2\pi$ on  $\mathbb R^2\setminus\cup_j\{\av_j\}$ such that 
\[\nabla\phi=-\nabla^\bot h+\Ab_0.\]
More precisely, we define $\phi$ as
\[ \phi(\xv)=\int_{\gamma(\xv_*,\xv)}(-\nabla^\bot h+\Ab_0)\cdot \diff\xv\]
where $\xv_*$ is a fixed point in $\R^2\setminus\cup_j\{\av_j\}$ and $\gamma(\xv_*,\xv)$ is a path in $\R^2\setminus\cup_j\{\av_j\}$ joining $\xv_*$ and $\xv$. If $\gamma'(\xv_*,\xv)$ is another path joining $\xv_*$ and $\xv$, then
\[ \int_{\ell(\xv_*,\xv)}(-\nabla^\bot h+\Ab_0)\cdot\diff\xv-\int_{\ell'(\xv_*,\xv)}(-\nabla^\bot h+\Ab_0)\cdot\diff\xv\in 2\pi\Z,\]
where $\gamma$ is the loop defined by $\gamma(\xv_*,\xv)\cup \gamma'(\xv_*,\xv)$ and oriented counter-clockwise.
Consequently, the function $\ee^{\ii\phi}$ is single-valued and well defined on $\R^2\setminus\cup_j\{\av_j\}$.
Furthermore, by periodicity of $h$ and linearity of $\Ab_0$,  there are real constants $\alpha_0$ and $\beta_0$ such that, modulo $2\pi$,  $\phi$ satisfies on $\R^2\setminus\cup_j\{\av_j\}$,
\[
    \phi(x_1+\sqrt{2\pi},x_2)=\phi(x_1,x_2)+\tfrac{\sqrt{2\pi}}{2} x_2+\alpha_0,\quad 
    \phi(x_1,x_2+\sqrt{2\pi})=\phi(x_1,x_2)-\tfrac{\sqrt{2\pi}}{2} x_1+\beta_0.
\]
Iterating this $N$ times and recalling that $R=\sqrt{2\pi}N$, we get
\[
    \phi(x_1+R,x_2)=\phi(x_1,x_2)+\tfrac{1}{2} Rx_2+\alpha,\quad 
    \phi(x_1,x_2+R)=\phi(x_1,x_2)-\tfrac{1}{2} Rx_1+\beta.
\]
with $\alpha=\frac{1}{\sqrt{2\pi}} R\alpha_0,$ $\beta=\frac{1}{\sqrt{2\pi}}R\beta_0$, and $\alpha_0,\beta_0$ are independent of $R$.

For $\xv\in Q_1$, we set $\rho(\xv)=\min(1,|\xv-\av_1|/\sqrt{b})$ and we extend $\rho$ by periodicity on all of $\R^2$. Then we introduce the function  $v$ defined on $K_R$ as 
\[v(\xv)=\begin{cases}\rho(\xv)\ee^{\ii\phi(\xv)},&\mbox{if }\xv\not\in\cup_j\{\av_j\},\\
0,&\mbox{if }\xv\in\cup_j\{\av_j\},
\end{cases}\]
and we notice that $v$ can be extended naturally to a function in $\mathcal{H}_R^{\alpha,\beta}$, thanks to the periodicity of the functions $h$ and $\rho$, the linearity of the vector potential $\Ab_0$, and the periodic conditions satisfied by $\phi$. By \cref{rem:g-ub} and \eqref{eq: u and v}, we assign to $v$ a function $u\in \mathcal{H}_R$ with $\mathcal G_{b,R}(u)= \mathcal G_{b,R}^{\alpha,\beta}(v)$.

Next, we proceed in estimating $\mathcal G_{b,R}^{\alpha,\beta}(v)$. By \eqref{eq:est-h-log-b} and the fact that $ R^{-1} |\mathbf{k}_{\alpha,\beta}| = O(1) $, 
\begin{multline*}
    \int_{Q_1}\Bigl(b|(\nabla-\ii\Ab_0-\ii R^{-1}\mathbf k_{\alpha,\beta})v|^2-|v|^2+\frac12|v|^4 \Bigr)\diff\xv\\
    =\int_{Q_1\setminus B(\av_1,\sqrt{b})}b|\nabla h|^2\diff\xv-\frac12|Q_1|+\frac12\int_{Q_1}(1-|v|^2)^2\diff\xv+\cO_{b}(b)\\
    =2\pi b|\log\sqrt{b}|-\pi+\cO_{b}(b),
\end{multline*}
and by periodicity we have
\[
\begin{aligned}
    \cG_{b,R}^{\alpha,\beta}(v)&=N^2\int_{Q_1}\Bigl(b|(\nabla-\ii\Ab_0-\ii R^{-1}\mathbf k_{\alpha,\beta})v|^2-|v|^2+\frac12|v|^4 \Bigr)\diff\xv\\
    &= \lf(b|\log\sqrt{b}|-\tfrac12+\cO_{b}(b) \ri)|K_R|, 
\end{aligned}
\]
which yields the formula in \eqref{eq:proof-g-ub*}, since $\mathcal G_{b,R}(u)= \mathcal G_{b,R}^{\alpha,\beta}(v)$.
\end{proof}
\subsection{Lower bound}
We prove a lower bound on $g'(b)$ that matches with the upper bound in \cref{prop:g-ub}. 
\begin{proposition}\label{prop:g-lb}
As $b\to 0_+$, we have that
\beq
	g'(b)\geq -\tfrac12\log b+o_{b}(\log b).
\eeq 
\end{proposition}
\begin{proof}
Take $b\in(0,1)\setminus \cJ$  and $\varepsilon>0$ small enough such that $b+\varepsilon\in(0,1)$. With $u_*$ any ground state of $\cG_{b,R}$, we have
\[\begin{aligned}
    0&\leq \bigl(g(b+\varepsilon)-g(b)\bigr)|K_R|+o(|K_R|)\\
    &\leq \cG_{b+\varepsilon,R}(u_*)-\cG_{b,R}(u_*)=\varepsilon\int_{K_R}|(\nabla-\ii\Ab_0)u_*|^2\diff\xv.
\end{aligned}\]
Using  \eqref{eq:asy-p} and  taking $R\to+\infty$ after dividing by $|K_R|$ with fixed $b$ and $\varepsilon$, we get
\begin{equation}\label{eq:lip-con-g}
0\leq g(b+\varepsilon)-g(b)\leq \varepsilon g'(b).
\end{equation}
Thanks to \eqref{eq:asym-g}, we have
\[ \zeta(b):=\frac{g(b)+\frac12+\frac{b}{2}\log b}{b\log b}=o_{b}(1).\]
We choose $\varepsilon=\delta b$ and $b\in(0,b_0)$, where $\delta$ and $b_0$ are chosen in $(0,1)$, so that $(1+\delta)b_0<1$. Writing
\[\begin{gathered}
g(b+\varepsilon)=-\tfrac12-\tfrac{(1+\delta)b}{2}\log \lf[(1+\delta)b \ri]
+\lf((1+\delta)b\log\bigl[(1+\delta)b\bigr] \ri)\zeta \lf((1+\delta)b \ri),\\
g(b)=-\tfrac12-\tfrac{b}{2}\log b+(b\log b)\zeta(b),
\end{gathered}\]
we get
\[g(b+\varepsilon)-g(b)=-\tfrac{\delta b}{2}\log b-\tfrac{(1+\delta)b}{2}\log(1+\delta)
+\mathcal R(\delta,b),\]
where
\[ \mathcal R(\delta, b)=(1+\delta )b \lf(\log b+\log(1+\delta) \ri)\zeta(b+\delta b)-(b\log b)\zeta(b),\]
Thus, 
\[\frac{g(b+\varepsilon)-g(b)}{\varepsilon|\log b|}\geq\frac12 -\frac{1+\delta}{2\delta|\log b|}\log(1+\delta)-\frac{1+\delta}{\delta}\Bigl(1+\frac{\log(1+\delta)}{|\log b|} \Bigr)\zeta(b+\delta b)+\frac{1}{\delta}\zeta(b).  \]
Using \eqref{eq:lip-con-g}, we get
\[\frac{g'(b)}{|\log b|}\geq\frac12 -\frac{1+\delta}{2\delta|\log b|}\log(1+\delta)-\frac{1+\delta}{\delta}\Bigl(1+\frac{\log(1+\delta)}{|\log b|} \Bigr)\zeta(b+\delta b)+\frac{1}{\delta}\zeta(b).\]
Taking $b\to0_+$ with fixed $ \delta $, we get
\[\liminf_{b\to0_+}\frac{g'(b)}{|\log b|}\geq \frac12,\]
which finishes the proof.
\end{proof}

\subsection{Consequences on the GL minimization}

We provide here the proofs of \cref{corol:potential,corol:density}.

\begin{proof}[Proof of \cref{corol:potential}]
Combining \cref{thm:g(b)}, \eqref{eq:bulk} and \eqref{eq:bulk*}, it is straightforward to estimate
\[ \frac{1}{|\Omega|} \int_{\Omega}(1-|\Psi_*|^2)^2\diff\xv = 1-2bg'(b)+2g(b)+o_\varepsilon(1) =o_b(b|\log b|) + o_\varepsilon(1),\]
and taking the $ \limsup $ as $ \varepsilon \to 0$, we get the result.
\end{proof}

\begin{proof}[Proof of \cref{corol:density}]
    It suffices to prove that for every $ b_* \in(0,1) $, there is $ b \in (0, b_*)$ and a minimizing configuration  $(\Psi_*,\Ab_*)_{b,\varepsilon}$ such that $|\Psi_*|$ does not converge pointwise as $\varepsilon\to0$ to a constant function, even along a subsequence.
    
    We argue by contradiction. If there is $b_*\in(0,1)$ such that for all $ b \in (0,b_*) $, $ |\Psi_*| $ converges to a constant profile along a subsequence, then thanks to \eqref{eq:bulk} and \eqref{eq:bulk*}, we would get
    \[bg'(b)-2g(b) = \sqrt{-2g(b)}\quad\mbox{a.e. on }(0,b_*).\]
    Solving the above differential equation with the initial condition $g(0)=-\frac12$, we find
    \[ g(b)=-\tfrac12(b-1)^2\]
    which violates \eqref{eq:asym-g} for $ b $ small enough (and \cref{thm:g(b)}).
\end{proof}

\section{Two scale vortex construction}\label{sec:vortex-ball}

\subsection{Main theorem}
Suppose that $\ell=\ell(\varepsilon)$ is positive and satisfies
\[\lim_{\varepsilon\to0}\varepsilon^{-1}\ell=+\infty\qquad\mbox{and}\qquad \lim_{\varepsilon\to0}\ell=0, \]
which we abbreviate as $\varepsilon\ll \ell\ll 1$. Later we will make further assumptions on $ \ell $, which will restrict the possible choices of the quantity (see next \eqref{eq:quantization}).
Consider $\xv_0\in\Omega$ such that 
\begin{equation}\label{eq:def-Q}
Q_\ell(\xv_0)=\xv_0+(-\ell/2,\ell/2)^2\subset \overline{Q_\ell(\xv_0)}\subset \Omega.
\end{equation}
We fix $b\in(0,1)$ and  a minimizing configuration $(\Psi_*,\Ab_*)_{b,\varepsilon}$ of the GL function $\cE_{b,\varepsilon}$.
Our aim is to construct vortex balls in the square $Q_{\ell}(\xv_0)$ and to estimate the sum of their degrees. In fact, we  prove the following
\begin{theorem}\label{thm:vortex-balls}
There exist constants \( b_0, \varepsilon_0 > 0 \) such that for all \( 0 < b \leq b_0 \) and \( 0 < \varepsilon < \varepsilon_0 \), the following holds:
\begin{enumerate}[\rm a)]
\item {\rm(existence)} the domain \( Q_\ell(\xv_0) \) contains disjoint vortex balls $\lf(B(\mathbf a_i,\varrho_i)\ri)_{i\in\mathcal I}$ with positive degrees $ d_i > 0 $;
\item {\rm(degrees)} let \( \mathfrak{D}_0 := \sum_{i \in \mathcal{I}} d_i \) and assume that \( \ell^2 \in 2\pi b^{-1}\varepsilon^2 \mathbb{N} \), then, as \( \varepsilon \to 0^+ \),  
   \[
   \left| \mathfrak{D}_0 - \frac{b\ell^2}{2\pi\varepsilon^2} \right| = b r_0(b) \mathcal O_{\eps}\left( \tfrac{\ell^2}{\varepsilon^2}\right) + o_{\eps}\left(\tfrac{\ell^2}{\varepsilon^2}\right),
   \]  
    where $r_0(b)$ is introduced in \eqref{eq:def-r0}, and the estimate holds  uniformly in \( b \in [b_1, b_0] \), for all $b_1\in(0,b_0)$.
   \end{enumerate}
\end{theorem}
Let us highlight that the constants $b_0$ and $\varepsilon_0$ are independent of $\ell$ and $\xv_0$. The rest of the section is devoted to the proof of \cref{thm:vortex-balls}.

\subsection{Preliminary estimates}\label{sec:asym}
As $\varepsilon\to0$, we have the following asymptotics (see, e.g., \cite[Prop. 4.2 \& 6.2]{At} or \cite[Prop. 3.1]{HK})
\begin{equation}\label{eq:asy-f}
\frac1{|Q_\ell(\xv_0)|}\int_{Q_\ell(\xv_0)}\Bigl(|(\nabla-\ii b\varepsilon^{-2}\Ab_*)\Psi_*|^2+\frac1{2\varepsilon^2}(1-|\Psi_*|^2)^2\Bigr)\diff \xv=\bigr(g(b)+\mbox{$\frac12$}+\mathrm r(\varepsilon,b)\bigl)\varepsilon^{-2},
\end{equation}
where, for any $b\in(0,1)$,
\[\lim_{\varepsilon\to0}\mathrm r(\varepsilon,b)=0,\]
and one can provide a quantitative estimate of the remainder $\mathrm r(\varepsilon,b)$ (see \cite[Eq.~(3.1)]{HK}. Specifically, for a given $b_1\in(0,1)$, we have
\[\mathrm r(\varepsilon,b)=\cO_{\eps}(\ell)+\cO_{\eps}(\varepsilon\ell^{-1})\]
uniformly with respect to $b$ in $[b_1,1)$.

Moreover, by \cite[Prop.~10.3.1 \& 11.4.4]{FH}, we have
\begin{equation}\label{eq:est-mf}
\begin{gathered}
\|\Psi_*\|_{L^\infty(\Omega)}\leq 1,\\ \|\curl\Ab_*-1\|_{C^1(\overline{\Omega})}=b^{-1/2}\cO_{\eps}(\varepsilon),\\
\|\curl\Ab_*-1\|_{C^2(\overline{\Omega})}=\cO_{\eps}(1),
\end{gathered}
\end{equation}
where the last two estimates are  uniform with respect to $b$ in any compact subinterval of $(0,1)$. In the sequel, we set $h_* :=\curl\Ab_*$. 
\subsection{Scaling}\label{sec:scaling}
We set
\begin{equation}\label{eq:def-R-KR}
R=\ell \sqrt{b}\,\varepsilon^{-1}\qquad\mbox{and}\qquad K_R=(-R/2,R/2)^2.\end{equation}
Notice that  $1\ll R\ll \varepsilon^{-1} $, i.e.
\[\lim_{\varepsilon\to0}R=+\infty\qquad\mbox{and}\qquad \lim_{\varepsilon\to0}\varepsilon R=0. \]
We impose the following quantization condition on $R$ 
\beq
	\label{eq:quantization}
	R^2\in 2\pi \mathbb N,
\eeq
which can in fact be obtained by carefully choosing $ \ell $ without loss of generality,
and we introduce the integer
\begin{equation}\label{eq:def-N}
N :=(2\pi)^{-1}R^2.
\end{equation}
To work in the square $K_R$, we perform a scaling and a translation by introducing
\begin{equation}\label{eq:def-scaling}
\psi(\xv):=\Psi_*\bigl(\xv_0+b^{-\frac12}\varepsilon\, \xv\bigr),\qquad \Ab(\xv):=b^{\frac12}\varepsilon^{-1}\Ab_*\bigl(\xv_0+b^{-\frac12}\varepsilon \,\xv\bigr),
\end{equation}
so that $ \xv\in K_R $ if the argument of the functions belongs to $ Q_{\ell}(\xv_0) $.

Then, the asymptotics in \eqref{eq:asy-f} reads as
\begin{equation}\label{eq:asy-f*}
\frac1{|K_R|}\int_{K_R}\Bigl(|(\nabla-\ii\Ab)\psi|^2+\frac1{2b}(1-|\psi|^2)^2 \Bigr)\diff\xv=\bigl(g(b)+\mbox{$\frac12$}\bigr)b^{-1}+o_R(1).
\end{equation}
Furthermore, setting
\begin{equation}\label{eq:def-h}
h=\curl\Ab,
\end{equation}
we infer from \eqref{eq:est-mf} that
\begin{equation}\label{eq:est-h}
\begin{gathered}
\|\psi\|_{L^\infty(K_R)}\leq 1,\\
\|h-1\|_{C(\overline{K_R})}=b^{-1/2}\cO_{\eps}(\varepsilon), \\
 \|\nabla h\|_{C(\overline{K_R})}=b^{-1}\cO_{\eps}(\varepsilon^2), \\
\|\Delta h\|_{C(\overline{K_R})}=b^{-1}\cO_{\eps}(\varepsilon^2),
\end{gathered}
\end{equation}
and we write the second equation in \eqref{eq:GL} as
\begin{equation}\label{eq:h}
-\nabla^\perp h=b^{-1}\varepsilon^2 \Real \lf\langle \ii\psi,(\nabla-\ii\Ab)\psi \ri\rangle_{\mathbb{C}}.
\end{equation}
\subsection{Vortex ball construction}\label{sec:vortex}

In order to prove \cref{thm:vortex-balls}, it suffices to construct the vortex balls for the configuration $(\psi,\Ab)$ in the square $K_R$ and establish the asymptotics of the sum of their positive degrees.

Let us decompose the square $K_R=(-R/2,R/2)^2$ into $N$ congruent squares $(Q_j)_{0 \leq j \leq N}$, such that $Q_j$ is congruent to $K_1=(-\sqrt{2\pi}/2,\sqrt{2\pi}/2)^2$. We then define
\[ \begin{aligned}
G(\psi;Q_j)&=\int_{Q_j} \Bigl( |(\nabla-\ii\Ab)\psi|^2+\frac1{2b}(1-|\psi|^2)^2\Bigr)\diff\xv,\\
    F(\psi,\Ab;Q_j)&=\int_{Q_j} \Bigl( |(\nabla-\ii\Ab)\psi|^2+|\curl\Ab|^2+\frac1{2b}(1-|\psi|^2)^2\Bigr)\diff\xv,\\
    F(u;Q_j)&=F(u,0;Q_j)=\int_{Q_j} \Bigl( |\nabla u|^2+\frac1{2b}(1-|u|^2)^2\Bigr)\diff\xv.
    \end{aligned}\]
Let $C^*>\pi$ be a  constant. We will make a definite choice of $C^*$ later on. 

We now divide the squares $ (Q_j)_{j} $ into two groups, which we name {\it good} and {\it bad} squares depending whether a suitable upper bound on the energy holds or not. This kind of division is reminiscent of arguments used in vortex analysis for Bose–Einstein condensates and GL theory \cite{1,6}. More precisely,  $Q_j$ is a {\it good} square if 
\begin{equation}\label{eq:def-good}
G(\psi,Q_j)\leq C^*|\log b|,
\end{equation}
while it is said to be {\it bad } otherwise. Let $N_{\mathrm{good}}$ be the number of good squares; consequently, $N_{\mathrm{bad}}:=N-N_{\mathrm{good}}$ is the number of bad squares. 

    We can find a rough  lower bound on $N_{\mathrm{good}}$ as follows: Given  a constant
    \[0<c<1-\frac{\pi}{C^*},\]
    there exists $b_0<1$, so that  for all $b<b_0$, there exists $R_0>0$ such that for $R\geq R_0$, we have
    \beq
    	\label{eq: Ngood lb}
    	N_{\mathrm{good}}\geq cN=cR^2/2\pi.
    \eeq
    This is obtained by exploiting the energy lower bound inside bad squares, which yields
    \[ C^*|\log b|N_{\mathrm{bad}} \leq \sum_{j \in J_{\mathrm{bad}}}G(\psi,Q_j) \leq \int_{K_R}\Bigl(|(\nabla-\ii\Ab)\psi|^2+\frac1{2b}(1-|\psi|^2)^2 \Bigr)\diff\xv,\]
   where $ J_{\mathrm{bad}} = \lf\{ j \in \lf\{1, \dots, N \ri\}: Q_j \mbox{ is bad} \ri\} $, and by using \eqref{eq:asy-f*} and \eqref{eq:asym-g}.

\begin{proposition}\label{prop:vortex-const}
Let $\alpha\in(0,1)$. There exist $b_0=b_0(\alpha)\in(0,1)$ and $C_0=C_0(\alpha)>0$ such that, for  $b<b_0$,  we have the following:
\begin{enumerate}[\rm i.]
\item if $Q_j$ is a good square, and $Q_j^b :=\{x\in Q_j\colon \mathrm{dist}(x,\partial Q_j)>\sqrt{b}\}$, then there exists a finite collection $\cB=\lf(B(\av_i,r_i) \ri)_{i \in \mathcal{I}}$  of disks such that
\[ \lf\{ \xv \in Q_j^b\colon \bigl|\,|\psi(\xv)|-1\,\bigr|\geq b^{\alpha/8} \ri\}\subset V_j:=Q_j^b\cap\cup_{i \in \mathcal{I}}B(\av_i,r_i),\]
and  $r(\cB):=\sum_{i \in \mathcal{I}} r_i$ satisfies $r(\cB)=|\log b|^{-2}$.
\item setting $B_i=B(a_i,r_i)$ for short and
	\beq
		d_i : =
		\begin{cases}
			\mathrm{deg}(\psi/|\psi|,B_i), &	\mbox{if } B_i\subset Q_j^b,	\\
			0,							&	\mbox{otherwise},
		\end{cases}
	\eeq
	 we have, 
\[ \int_{V_j}\Bigl(|(\nabla-\ii\Ab)\psi|^2+\frac1{2b}(1-|\psi|^2)^2+ \frac{1}{|\log b|^{4}}|\curl\Ab|^2 \Bigr)\diff\xv\geq 2\pi D\Bigl(\log \frac{r(\cB)}{D\sqrt{b}}-C_0\Bigr) .\]
\end{enumerate}
\end{proposition}

\begin{remark}\label{rem:constants}
We will see in the proof of \cref{prop:vortex-const} that $b_0$ depends on $\alpha$ and the constant $C^*$ in the definition of the good squares. However, the constant $C_0$ can  be chosen independently of $\alpha$ and $C^*$.
\end{remark}

\begin{proof}[Proof of Proposition~\ref{prop:vortex-const}]
With 
\begin{equation}\label{eq:def-epsilon}
\epsilon =\sqrt{b},
\end{equation}
we read \eqref{eq:def-good} as
\begin{equation}\label{eq:def-good*}
\int_{Q_j}\Bigl(|(\nabla-\ii\Ab)\psi|^2+\frac1{2\epsilon^2}(1-|\psi|^2)^2 \Bigr)\diff \xv\leq C^*|\log \epsilon|.
\end{equation}
We now choose $ \epsilon_1  $ small enough in order that for $ \epsilon < \epsilon_1 $ the assumptions of  \cite[Thm.~4.1]{SS07} are satisfied: let us then take $ \epsilon_1 > 0 $ so that for $ \epsilon <\epsilon_1 $, we have
\[C^*|\log \epsilon| < \tfrac{1}{2} \epsilon^{\alpha-1}.\]
Now, by the diamagnetic inequality,
\[
\int_{Q_j}\Bigl(|\nabla|\psi||^2+\frac1{2\epsilon^2}(1-|\psi|^2)^2 \Bigr)\diff \xv \leq \epsilon^{\alpha-1}.
\]
Hence, we can apply \cite[Thm.~4.1]{SS07} and obtain that for any $ \alpha \in (0,1) $ there exists $ \epsilon_0(\alpha)>0 $ so that the result holds for any $ \epsilon < \epsilon_0(\alpha) $ and for any $ r(\mathcal{B}) $ satisfying
\bdm
	C \epsilon^{\alpha/2} < r(\mathcal{B}) < 1,
\edm
where $ C $ is a universal constant. We want now to choose $ r(\mathcal{B}) = (2|\log\epsilon|)^{-2} $ and therefore we have to require that $ \epsilon $ is smaller than yet another threshold $ \epsilon_2 > 0  $, satisfying
\bdm
	C \epsilon_2^{\alpha/2} \leq \frac{1}{4 |\log\epsilon_2|^2} \leq 1.
\edm
By setting $ \epsilon_0 : = \min \lf\{ \epsilon_1, \epsilon_0(\alpha), \epsilon_2 \ri\} $, all the conditions are met for $ \epsilon < \epsilon_0 $ and therefore the result holds for any $ b_0 $ and $ C_0 $ chosen arbitrarily in $ (0, \epsilon_0^2] $ and $ (0, C] $, respectively.
\end{proof}
\begin{remark}\label{rem:vortex}
 Fixing $b_1$ in $(0,b_0)$, there is a  constant $\varepsilon_0$ depending on $b_1$ and $b_0$ so that the estimate in \eqref{eq:est-h} yields
\[\int_{Q_j}|\curl\Ab|^2\,\diff\xv\leq C^*|\log\sqrt{b}|, \qquad\mbox{for  $0<\varepsilon<\varepsilon_0$ and $b_1\leq b\leq b_0$}.\]
We then infer the stronger bound
\[\int_{Q_j}\Bigl(|(\nabla-\ii\Ab)\psi|^2+|\curl\Ab|^2+\frac1{2\epsilon^2}(1-|\psi|^2)^2 \Bigr)\diff \xv \leq 2C^*|\log\epsilon|.\]
Consequently, applying again \cite[Thm.~4.1, item (4)]{SS07}, we find the upper bound on $D$
\beq
	D\leq C\frac{2C^*|\log\epsilon|}{\alpha|\log \epsilon|}=\frac{2C^*C}{\alpha}, 
\eeq
for a universal constant $C$. This proves that the degree is bounded, provided that $\varepsilon$ is small enough, or, equivalently, the square $K_R$ has sufficiently large area, i.e., $R\geq R_0$ for some $R_0=R_0(b_1,b_0)$. Thanks to \cref{prop:vortex-const}, we  deduce that  
\[ G(\psi,Q_j)\geq  \pi D|\log b|-\widetilde{C} \log|\log b|-\widetilde{C},\]
for a positive constant $\widetilde{C} $ that depends on $\alpha$ and $C^*$.
\end{remark}
If $ J_{\mathrm{good}} $ is the index set of good square, applying \cref{prop:vortex-const} with $\alpha=1/2$ to every $Q_j$, $ j \in  J_{\mathrm{good}} $, we get disks 
\[ \lf(B_i^{(j)}=B\lf(\av_i^{(j)},r_i^{(j)}\ri)\ri)_{i \in \mathcal{I}_j}, \qquad \mbox{with degrees } \lf(d_i^{(j)}\ri)_{i \in \mathcal{I}_j}.
\]
Setting 
\[
	D^{(j)} := \sum_{i \in \mathcal{I}_j} \lf|d_i^{(j)} \ri|,
\]
we also know that $D^j$ is bounded and 
\[ V_j=Q_j^b\cap B_i^j\supset \lf\{\xv \in Q_j^b\colon \lf||\psi(\xv)|-1\ri |\geq b^{1/16} \ri\}.\]

We now prove a lower bound on the total degree of the vortex balls inside a good cell.

\begin{proposition}\label{prop:estimate-degree}
There exists $b_2\in(0,1)$, so that, for every $b_1 \in(0,b_2)$, there exist positive constants $C$ and $\widehat C$ independent of $C^*$ in \eqref{eq:def-good} such that, if $b\in(b_1,b_2)$, we have
\begin{equation}\label{eq:lb-vortex}
\sum_{j\in J_{\mathrm{good}}} D^{(j)} \geq N_{\mathrm{good}}-C|\log b|^{-1/2}N_{\mathrm{good}}-\widehat CN_{\mathrm{bad}}+o_{b}(N_{\mathrm{good}}).
\end{equation}
\end{proposition}
\begin{proof}
Let us consider the vorticity measure {(defined in \eqref{eq: vorticity} modulo scaling),} i.e.,
\[ \mu=\curl(\ii\psi,(\nabla-\ii\Ab)\psi)+\curl\Ab\]
and notice that by \eqref{eq:h} 
\[\mu=-b\varepsilon^{-2}\Delta h+h. \]
With $\epsilon=\sqrt{b}$, $M=F(\psi,\Ab;Q_j)$ and $r=r(\cB^{(j)})=|\log b|^{-2}$, we have by \cite[Thm.~6.1]{SS07}: 
\beq
	 \lf\|\mu-2\pi\sum_{i}d_i^{(j)}\delta_{\av_i^{(j)}} \ri\|_{{C_0^{0,1}(Q_j)}^*}\leq
C\max(r,\epsilon)(M+1),
\eeq
where $C$ is a universal constant, $Q_j$ is a good square, and recall that we set $d_i^j=0$ if $B_i^j\not\subset Q_j$. The space $C_0^{0,1}(Q_j)$ consists of the Lipschitz functions on $Q_j$ with compact support, endowed with the Lipschitz norm \eqref{eq: Lipschitz}.

By \eqref{eq:est-mf} and \eqref{eq:def-good*}, we have
\[ M = \cO_{\epsilon}(|\log \epsilon|)=\cO_{b}(|\log b|^{-1}).\]
Consequently
\begin{equation}\label{eq:jac-est}
\lf\|\mu-2\pi\sum_{i}d_i^{(j)}\delta_{\av_i^{(j)}} \ri\|_{{C_0^{0,1}(Q_j)}^*} = 
\cO_{b}(|\log b|^{-1}).
\end{equation}
Choose a function $\chi\in C_c^\infty\bigl((-\sqrt{\pi/2},\sqrt{\pi/2}\,)^2;\R\bigr)$ such that 
\[\|\chi\|_\infty\leq 1\qquad\mbox{and}\qquad\chi=1\mbox{ on } \lf(-\sqrt{\tfrac\pi2}+\delta,\sqrt{\tfrac\pi2}-\delta \ri)^2, \]
with $\delta$ fixed such that $0<\delta \ll 1$. Notice that \[\|\nabla\chi\|_\infty \leq C \delta^{-1}, \qquad \|\chi\|_{\rm Lip}= C \delta^{-1}.\]
We introduce the function $\varphi_j\in C_c^\infty(Q_j)$ defined as
\[\varphi_j(\xv)=\chi(\mathbf x-\mathbf q_j)\]
where $\mathbf q_j$ denotes the center of the square $Q_j$. Notice that $\varphi_j=1$ on 
\[\widehat Q_j:=\mathbf q_j+ \lf(-\sqrt{\pi/2}+\delta,\sqrt{\pi/2}-\delta \ri)^2.\]
Finally, we introduce the function $\varphi$ defined on $K_R$ as
\beq
	\label{eq: varphi}
	\varphi(\xv)=\begin{cases}
\varphi_j(\xv),&\mbox{if }\xv\in Q_j\mbox{ and }Q_j\mbox{ is a good square,}\\
0,&\mbox{otherwise.}
\end{cases}
\eeq
We have
\begin{equation}\label{eq:vorticity-measure} \int_{K_R}\mu\,\varphi\,\diff\xv=-b\varepsilon^{-2}\int_{K_R}
\Delta h\,\varphi\,\diff \mathbf x+\int_{K_R}
h\,\varphi\,\diff \mathbf x.
\end{equation}
By \eqref{eq:est-h}, we write
\begin{equation}\label{eq:est-h-phi} \int_{K_R}h\,\varphi\,\diff\xv
    =2\pi N_{\mathrm{good}}+N_{\mathrm{good}}\cO( \delta )+b^{-1/2}|K_R|\cO_{\eps}(\varepsilon ).
\end{equation}

In next \cref{lem: Delta h}, we prove that
\bdm
\left|\int_{K_R}\Delta h\,\varphi\,\diff\mathbf x\right|\leq \widehat Cb^{-1}\varepsilon^2(N_{\mathrm{good}} \delta+N_{\mathrm{bad}}+R)
\edm
where $\widehat C$ is a constant, independent of $ b, \varepsilon, R, d$ and also of the constant $C^*$ introduced in \eqref{eq:def-good}. 
Inserting the estimates \eqref{eq:est-h-phi} and \eqref{eq:est-Delta h-phi} into \eqref{eq:vorticity-measure}, we get
\[\left|\int_{K_R}\mu\,\varphi\,\diff\xv-2\pi N_{\mathrm{good}}\right|\leq N_{\mathrm{good}}\cO( \delta )+b^{-1/2}|K_R|\cO_{\eps}(\varepsilon )+\cO_{\eps}(R)+\widehat CN_{\mathrm{bad}} \]
and, combining it with \eqref{eq:jac-est}, we find
\[ \left|\int_{K_R}\mu\,\varphi\,\diff\xv-2\pi\sum_{i,j}d_i^{(j)}\varphi \lf(\av_i^{(j)} \ri) \right|
= N_{\mathrm{good}} \delta^{-1} \mathcal O_{b} \lf( |\log b|^{-1} \ri).\]
Consequently,
\bmln{
	 \lf| \sum_{i,j} d_i^{(j)} \varphi\lf(\av_i^{(j)}\ri) - N_{\mathrm{good}} \ri|\leq  C \delta  N_{\mathrm{good}} +|K_R| \cO_{\eps} (\varepsilon )+\cO_{\eps}(R) \\
	 +\delta^{-1} N_{\mathrm{good}} \mathcal O_{b} \lf( |\log b|^{-1} \ri)+\widehat CN_{\mathrm{bad}},
	 }
uniformly with respect to $ \delta \ll 1$ and $b\in(b_1,b_2)$. We choose now $ \delta =|\log b|^{-1/2}$ and $b_2$ sufficiently small so that, for\footnote{Once we choose $b_2$, we pick any $b_1\in(0,b_2)$ and fix it afterwards.} $b\in(b_1,b_2)$, the above estimate becomes
\beq
	 \lf| \sum_{i,j} d_i^{(j)} \varphi\lf(\av_i^{(j)}\ri) - N_{\mathrm{good}} \ri| = N_{\mathrm{good}} \mathcal O_{b} \lf( |\log b|^{-1/2} \ri)  +|K_R| \cO_{\eps} (\varepsilon )+\cO_{\eps}(R).
\eeq
By \eqref{eq: Ngood lb}, $N_{\mathrm{good}}$ grows at least like $R^2$. Recalling that $ D^{(j)} = \sum_i |d_i^{(j)} | $ and
\bdm
	d_i^{(j)} \varphi\lf(\av_i^{(j)} \ri) \leq \max\lf(d_i^{(j)},0 \ri), 
\edm
we eventually obtain the desired lower bound,
\begin{equation}\label{eq:con-est-sum-D}
\sum_jD^j\geq \sum_{i,j} \max\lf(d_{i}^{(j)},0 \ri) \geq \lf( 1 +  \cO_{b}\lf( |\log b|^{-1/2} \ri) - \widehat C  + o_{\eps}(1) \ri)N_{\mathrm{good}} .
\end{equation}
\end{proof}

\begin{lemma}\label{lem: Delta h}
	Let $ h $ and $ \varphi $ be defined in \eqref{eq:h} and \eqref{eq: varphi}, respectively. Then, there exists a constant $\widehat C$ independent of $ b,\varepsilon,R, \delta $ and $ C^* $  introduced in \eqref{eq:def-good}, such that
	\begin{equation}\label{eq:est-Delta h-phi}
\left|\int_{K_R}\Delta h\,\varphi\,\diff\mathbf x\right|\leq \widehat Cb^{-1}\varepsilon^2(N_{\mathrm{good}} \delta+N_{\mathrm{bad}}+R)
\end{equation}
\end{lemma}
\begin{proof}
Since $\varphi=0$ over bad squares, we have
\[\int_{K_R}\Delta h\,\varphi\,\diff\xv=\sum_{j\in J_{\mathrm{good}}}\int_{Q_j}
\Delta h\,\varphi\,\diff \mathbf x. \]
The decomposition $\varphi=1+(\varphi-1)$ yields, for every good square $Q_j$,
\[\int_{Q_j}
\Delta h\,\varphi\,\diff \mathbf x=\int_{Q_j}
\Delta h\,\diff \mathbf x+\int_{Q_j}
\Delta h\,(\varphi-1)\diff \mathbf x.\]

Note that the estimates below do not use the condition \eqref{eq:def-good} defining good squares; they rely only on the properties of $h$ and the Lipschitz regularity of $\varphi$. Hence the constants  can be chosen independent of $C^*$, including the constant $\widehat C$.

Let $\xv^*\in \overline{\widehat Q_j}$ be the nearest point to $\xv\in Q_j\setminus \widehat Q_j$; for instance, $|\xv-\xv^*|=\cO(\delta)$ and $\varphi(\xv_*)=1$ if $Q_j$ is a good square. Knowing that $\varphi=\varphi_j$ on $Q_j$ and that $\varphi_j=1$ on $\widehat Q_j$, we write
\[\begin{aligned}
\int_{Q_j}
\Delta h\,(\varphi-1)\diff \mathbf x
\leq \|\Delta h\|_\infty \int_{Q_j\setminus\widehat Q_j}\|\varphi_j\|_{\mathrm{Lip}(Q_j)}|\xv-\xv^*|\,\diff \xv\leq  Cb^{-1} \delta \varepsilon^2,
\end{aligned}\]
where we used \eqref{eq:est-h} to estimate $\|\Delta h\|_\infty$, combined with  the following properties
\[\|\varphi\|_{\mathrm{Lip}(Q_j)}=\cO(\delta^{-1}),\quad |Q_j\setminus \widehat Q_j|=\cO(\delta),\quad |\xv-\xv_*|=\cO(\delta) \mbox{ for }\xv\in Q_j\setminus\widehat Q_j. \]
Summing over the good squares, we get
\[ \sum_{j\in J_{\mathrm{good}}}\int_{Q_j}\Delta h\,(\varphi-1)\diff \mathbf x= b^{-1} N_{\mathrm{good}} \delta \cO_{\eps}(\varepsilon^2).\]

Estimating the sum of the integral of $\Delta h$ over the good squares is more subtle. It is based on the integration by parts formula
\[ \int_{Q_j}\Delta h\,\diff\xv=\int_{\partial Q_j}\frac{\partial h}{\partial \nu_j}\diff s(\xv),\]
where $\nu_j$ is the normal on the smooth parts of $\partial Q_j$. We then observe that $\partial Q_j$ consists of four edges, and if an edge $\Gamma$ is  also an edge of a good square $Q_{j'}$, then 
\[\int_{\Gamma}\frac{\partial h}{\partial \nu_j}\diff s(\xv)+\int_{\Gamma}\frac{\partial h}{\partial \nu_{j'}}\diff s(\xv)=0. \]
If $\partial Q_j$ has an edge $\Gamma$ adjacent to a bad square or it lies on an edge of $K_R$, we use \eqref{eq:est-h} and write
\[ \int_{\Gamma}\frac{\partial h}{\partial \nu_j}\diff s(\xv)=b^{-1}\mathcal O_{\eps}( \varepsilon^2).\]
Since the number of bad squares is $N_{\mathrm{bad}}$, and the number of squares with some edge lying on $\partial K_R$ is $\mathcal O_{\eps}(R)$, we conclude that
\[ \left|\sum_{j\in J_{\mathrm{good}}}\Delta h\,\diff\xv\right| = b^{-1} N_{\mathrm{bad}} \cO_{\eps}(\varepsilon^2) + b^{-1} \cO_{\eps}(\varepsilon^2 R).\]
\end{proof}
If we introduce the index sets
\begin{eqnarray}
    \mathcal{J}_{\mathrm{good}}^+ &=& \lf\{(i,j)\colon j\in J_{\mathrm{good}}, d_i^{(j)} \geq 1 \ri\}, \nonumber\\
    \mathcal{J}_{\mathrm{good}}^0 &=& \lf\{(i,j)\colon j\in J_{\mathrm{good}},d_i^{(j)} =0 \ri\},\nonumber\\ 
    \mathcal{J}_{\mathrm{good}}^- &=& \lf\{(i,j)\colon j\in J_{\mathrm{good}}, d_i^{(j)} \leq -1 \ri\}.\nonumber
\end{eqnarray}
We can rewrite \eqref{eq:con-est-sum-D} as
\[\sum_{(i,j)\in  \mathcal{J}_{\mathrm{good}}^+} d_i^{(j)} \geq N_{\mathrm{good}} \lf( 1 +\cO_{b} \lf(|\log b|^{-1/2}\ri) + o_{\eps}(1) \ri) -\widehat CN_{\mathrm{bad}}.\]
Let $\mathcal{N}_{\mathrm{good}}^+$, $\mathcal{N}_{\mathrm{good}}^0$   and $\mathcal{N}_{\mathrm{good}}^-$ be the cardinality of $\mathcal{J}_{\mathrm{good}}^+$, $\mathcal{J}_{\mathrm{good}}^0$
and $\mathcal{J}_{\mathrm{good}}^-$, respectively.
By \cref{rem:vortex}, $|d_i^{(j)}|$ is bounded uniformly with respect to $(i,j)$, hence there is a constant $\widetilde C\geq 1$ such that
\[\mathcal{N}_{\mathrm{good}}^+\leq \sum_{(i,j) \in \mathcal{J}_{\mathrm{good}}^+}d_i^{(j)} \leq \widetilde C\mathcal{N}_{\mathrm{good}}^+,\qquad \mathcal{N}_{\mathrm{good}}^-\leq \sum_{(i,j)\in  \mathcal{J}_{\mathrm{good}}^-} \lf| d_i^{(j)} \ri| \leq \widetilde C \mathcal{N}_{\mathrm{good}}^-.\]
Consequently,
\beq
	\label{rem:di=1}
	\sum_{j\in J_{\mathrm{good}}} D^{(j)}=\sum_{(i,j)\in \mathcal{J}_{\mathrm{good}}^+} d_i^{(j)} + \sum_{(i,j)\in \mathcal{J}_{\mathrm{good}}^-} \lf| d_i^{(j)} \ri| \geq \sum_{(i,j) \in \mathcal{J}_{\mathrm{good}}^+} d_i^{(j)} + \mathcal{N}_{\mathrm{good}}^-.
\eeq

Now we prove that the number of good squares is dominant.
\begin{proposition}\label{prop:g-lb-op}
If $C^*$ in \eqref{eq:def-good} is chosen large enough, then the good squares cover almost all the square $K_R$, i.e., there exists a constant $C>0$ independent of $b\in[b_1,b_2]$ and $R$ such that, as $R\to+\infty$ (or equivalently as $ \eps \to 0 $),
\beq
	\lf(1-Cr_0(b)+o_{\eps}(1)\ri)|K_R|/2\pi\leq N_{\mathrm{good}}\leq |K_R|/2\pi,
\eeq
where $r_0(b)$ is introduced in \eqref{eq:def-r0}. 
\end{proposition}
\begin{proof}
It is obvious that $N_{\mathrm{good}}\leq N=|K_R|/2\pi$, so we just need to prove the lower bound. With \eqref{eq:lb-vortex} in hand, we get from \cref{prop:vortex-const} and \cref{rem:vortex} that
\begin{equation}\label{eq:G-lb-good}
\sum_{j\in J_{\mathrm{good}}} G(\psi;Q_j) \geq \lf( \pi|\log b|+\mathcal O_{b}(r_0(b)) + o_{\eps}(1) \ri) N_{\mathrm{good}} - \widehat CN_{\mathrm{bad}},
\end{equation}
where $r_0(b)$ is introduced in \eqref{eq:def-r0}.
Then, we deduce from \eqref{eq:asy-f*} and the definition of bad squares that
\[
\lf(\pi|\log b|+\mathcal O_{b}(r_0(b)) \ri) N_{\mathrm{good}} + \lf(C^*-\widehat C \ri)|\log b| N_{\mathrm{bad}}
\leq \lf(g(b)+\tfrac12 \ri)b^{-1}|K_R|+o_{\eps}(|K_R|),
\]
where we used that $|K_R|=2\pi N$ and that $N_{\mathrm{good}}$ grows like $R^2$. Choosing $C^*>\widehat C$, we get in view of \eqref{eq:asym-g} that
\[ |\log b|N_{\mathrm{bad}}\leq \lf( O_{b}(r_0(b))+o_{\eps} (1) \ri)\mathcal N , \]
which eventually yields that
\[N_{\mathrm{good}}\geq \lf(1+\mathcal O_{b}(r_0(b)) + o_{\eps} (1) \ri)N.\]
\end{proof}
In the sequel, we fix the choice of $C^*$ so that \cref{prop:g-lb-op} holds. We prove next that the lower bound in \cref{prop:estimate-degree} is optimal, in the sense that a matching upper bound holds.
\begin{proposition}
    \label{prop:ub-op}
   There exists a constant $C>0$ such that, for $b\in [b_1,b_2]$, we have, as $R\to+\infty$,
    \beq
    	\lf(1-Cr_0(b)+o_{\eps}(1) \ri) N_{\mathrm{good}} \leq \sum_{j\in J_{\mathrm{good}}} D^{(j)}
 \leq \lf(1+Cr_0(b)+o_{\eps}(1)\ri) N_{\mathrm{good}},
    \eeq
    where $r_0(b)$ is introduced in \eqref{eq:def-r0}.
\end{proposition}

\begin{proof}
In view of \cref{prop:estimate-degree,prop:g-lb-op}, we only need to prove the upper bound.  Firstly, we have by \cref{prop:vortex-const} and \cref{rem:vortex}, that
\[
\sum_{j\in J_{\mathrm{good}}} G(\psi;Q_j)\geq \lf(\pi|\log b|+\cO_{b}(r_0(b)) \ri)\sum_{j\in  J_{\mathrm{good}}}D^{(j)}.
\]
Secondly,  we have by  \eqref{eq:asy-f*} and \eqref{eq:asym-g},
\[
\sum_{j\in J_{\mathrm{good}}} G(\psi;Q_j)\leq G(\psi;K_R)\leq \pi |\log b| N+R^2\cO_{b}(r_0(b))+o_{\eps}(R^2).
\]
It remains to apply Proposition~\ref{prop:g-lb-op} to conclude.
\end{proof}

    It follows from \cref{prop:g-lb-op} and \cref{rem:di=1} that almost all the degrees are positive. Specifically, 
    \[\left|\sum_{(i,j)\in  \mathcal{J}_{\mathrm{good}}^+} d_i^{(j)} - N_{\mathrm{good}} \right| \leq  \lf(\cO_{b}(r_0(b))+o_{\eps}(1) \ri) N_{\mathrm{good}}, 
    \]
    \[
    \sum_{(i,j)\in  \mathcal{J}_{\mathrm{good}}^-} \lf| d_i^{(j)} \ri| = \lf(\cO_{b}(r_0(b))+o_{\eps}(1) \ri)N_{\mathrm{good}},  \]
    and the number $\mathcal{N}_{\mathrm{good}}^+$ of balls with positive degree satisfies
    \[\widetilde{C}\lf(1-\cO_{b}(r_0(b)) - o_{\eps}(1) \ri)  N_{\mathrm{good}} \leq \mathcal{N}_{\mathrm{good}}^+ \leq \lf(1+\cO_{b}(r_0(b))+o_{\eps}(1) \ri) N_{\mathrm{good}},\]
    whereas the number $\mathcal{N}_{\mathrm{good}}^-$ of balls with negative degree satisfies
    \[ \mathcal{N}_{\mathrm{good}}^- \leq \lf(\cO_{b} (r_0(b)) + o_{\eps}(1) \ri) \mathcal{N}_{\mathrm{good}}^+.\]   
\section{Vortex Balls and Proof of \cref{thm:main}}

In view of the two-scale vortex construction we did in \cref{sec:vortex-ball}, we are ready now to give  the proof of \cref{thm:main}. Consider a function $K(\varepsilon)$ such that
    \[\lim_{\varepsilon\to0^+}K(\varepsilon)=+\infty,\qquad \lim_{\varepsilon\to0^+} \varepsilon^2K(\varepsilon)=0.\]
    Let us introduce the integer $N(\varepsilon)=\lfloor K(\varepsilon)\rfloor$ and the parameters
    \[\ell=\varepsilon\sqrt{\frac{2\pi}{b} N(\varepsilon)},\qquad R=\ell\sqrt{b}\varepsilon^{-1}=\sqrt{2\pi N(\varepsilon)}.\]
    Then $\ell$ satisfies the hypothesis in \cref{sec:asym}, and the quantization condition in \eqref{eq:def-N} holds.

    Consider a lattice of squares $(Q_\ell(\xv_j))$ in $\R^2$, with fundamental cell as in \eqref{eq:def-Q}. We introduce the index set
    \[\mathcal I=\{j\colon Q_\ell(\xv_j)\subset\Omega\} \]
    and notice that the cardinality of $\mathcal I$ is equal to $|\Omega|\ell^{-2}+\mathcal O_{\eps}(\ell^{-1})$.
    In the square $Q_\ell(\xv_j)\subset\Omega$, we apply Theorem~\ref{thm:vortex-balls} and obtain a class $\cB_j$ of vortex balls in $Q_\ell(\xv_j)\subset\Omega$ with positive degrees. 
    We denote by $\mathfrak D^{(j)} $  the sum of these degrees. Then, by \cref{thm:vortex-balls}, we know that $\mathfrak D^{(j)}$ is equal to $(2\pi)^{-1}R^2=(2\pi)^{-1}\ell^2b\varepsilon^{-2}$, up to $(\cO_{b}(r_0(b))+o_{\eps}(1))R^2$ errors. Specifically,
    \[ \lf| \mathfrak D^{(j)}-(2\pi)^{-1}\ell^2b\varepsilon^{-2} \ri| \leq \lf(\cO_{b}(r_0(b))+o_{\eps}(1) \ri)\ell^2b\varepsilon^{-2}.\]
    Let us denote by $(B(\av_i,\varrho_i))_{i \in I} $ the union of all the vortex balls in $\cB_j$ over $ j $ with $j\in\mathcal I$. 
   The sum of their degrees, $d_i:=\mathrm{deg}(\Psi_*,\partial B(\av_i,\varrho_i)$, is
    \[\mathfrak D:=\sum_{i \in I}d_i=\sum_{j\in\mathcal I}\mathfrak D^{(j)}.\]
    Then, to leading order, $\mathfrak D$ is equal to $(2\pi)^{-1}b\varepsilon^{-2}$, up to $(\cO_{b}(r_0(b))+o_\varepsilon(1))\varepsilon^{-2}$ errors. 
    
    For a given $j\in\mathcal I$, we introduce the index set
    \[\mathcal I_j= \lf\{i \in I \colon B(\av_i,\varrho_i)\subset Q_\ell(\xv_j) \ri\}.\]
    In view of this notation,
    \[\sum_{i\in\mathcal I_j}d_i=\mathfrak D^{(j)}.\]
    Consider $\varphi\in C_0^{0,1}(\Omega)$ with Lipschitz norm
    \[\|\varphi\|_{C_0^{0,1}(\Omega)}=1.\]
    In $Q_\ell(\xv_j)$, we can approximate $\varphi$ by $\varphi(\xv_j)+\cO_{\eps}(\ell)$. Thus,    
    \[\begin{aligned}
        \int_\Omega\varphi(\xv)\diff\xv&=\sum_{j\in\mathcal I}\int_{ Q_j(\xv_j)}\varphi(\xv)\diff\xv+\mathcal O_{\eps}(\ell)\\
        &=\ell^2\sum_{j\in\mathcal I}\varphi(\xv_j)+\mathcal O_{\eps}(\ell),
    \end{aligned}\]
    and
    \[  \sum_{i\in\mathcal I_j}d_i\varphi(a_i)=\lf((\varphi(\xv_j)+\cO_{\eps}(\ell) \ri)\mathfrak D^{(j)}.
    \]
    Consequently,
    \[\begin{aligned}
        \left|b^{-1}\varepsilon^2\sum_{i \in I}2\pi d_i\varphi(a_i)-\int_\Omega\varphi(\xv)\diff \xv\right|&\leq \sum_{j\in\mathcal I} \lf|2\pi b^{-1}\varepsilon^2\mathfrak D^{(j)}-\ell^2 \ri| +\mathcal O_{\eps}(\ell)\\
        &= \cO_{b}(r_0(b))+o_\varepsilon(1).
    \end{aligned}
    \]
    Since the Lipschitz norm of $\varphi$ is equal to $1$, we obtain
    \[
    	\lf\|b^{-1}\varepsilon^2\sum_{i \in I}2\pi d_i\delta_{\av_i} - \mathscr{L} \ri\|_{{C_0^{0,1}(\Omega)}^*} = \cO_{b}(r_0(b)) +o_\varepsilon(1),\]
    which finishes the proof of \cref{thm:main}.

\bigskip

\begin{footnotesize}
\noindent
\textbf{Acknowledgments.} 
We would like to thank the referee for the useful comments, which led to an improvement of the paper.
A.K. acknowledges support from a startup fund at AUB. This work was started when A.K. was at The Chinese University of Hong Kong, Shenzhen, and he acknowledges support through grant no. UDF01003322 and UF02003322.
M.C. acknowledges support from PNRR Italia Domani and Next Generation Eu through the ICSC National Research Centre for High Performance Computing, Big Data and Quantum Computing, and from the MUR grant ``Dipartimento di Eccellenza 2023-2027'' of Dipartimento di Matematica, Politecnico di Milano.
\end{footnotesize}


\begin{thebibliography}{mich00}
%
\bibitem[Ab]{Ab} A. Abrikosov. On the magnetic properties of superconductors of the second type. {\it Soviet Phys. JETP} {\bf 5}:1174--1182 (1957).
%
\bibitem[Al]{Al} {Y. Almog.
Abrikosov lattices in finite domains.
{\it Commun. Math. Phys.} {\bf 262}:677--702 (2006).}
%
\bibitem[AH]{AH} {Y. Almog, B. Helffer.
The distribution of surface superconductivity along the boundary: on a conjecture of X. B. Pan.
{\it SIAM J. Math. Anal.} {\bf 38}:1715--1732 (2007).}
%
\bibitem[AS]{AS} A. Aftalion, S. Serfaty.
Lowest Landau level approach in superconductivity for the Abrikosov lattice close to $H_{c_2}$.
{\it Selecta Math.} {\bf 13}:183--202 (2007).
%
\bibitem[At]{At} K. Attar. The ground state energy of the two dimensional Ginzburg-Landau functional with variable
magnetic field. {\it Ann. Inst. H. Poincar\'e C -- Anal. Non Lin\'eaire} {\bf 32}:325--345 (2015).
%
\bibitem[ASa]{ASa} H. Aydi, E.  Sandier.
Vortex analysis of the periodic Ginzburg-Landau model. 
{\it Ann. Inst. H. Poincar\'e C -- Anal. Non Lin\'eaire} {\bf 26}:1223--1236 (2009).
%
\bibitem[BJOS]{BJOS}	 	S. Baldo, R.L. Jerrard, G. Orlandi, H.M. Soner. Vortex Density Models for Superconductivity and Superfluidity. {\it Commun. Math. Phys.} {\bf  318}:131--171 (2013).
%
\bibitem[BBH]{BBH}	F. Bethuel, H. Brezis, F. H\'{e}lein. {\it Ginzburg–Landau Vortices}.
Progr. Nonlinear Differential Equations Appl. {\bf 13}, Birkh\"auser Boston (1994).
%
\bibitem[C]{Cor}	M. Correggi. Surface Effects in Superconductors with Corners. {\it Bull. Unione Mat. Ital.} {\bf 14}:51--67 (2021).
%
\bibitem[CC]{CC}	A. Calignano, M. Correggi. Derivation of the Gross-Pitaevskii Theory for Interacting Fermions in a Trap. In: M. Correggi, M. Falconi (eds). Quantum Mathematics I. Springer INdAM Series, vol 57. Springer, Singapore (2023).
%
\bibitem[CR1]{CR1} M. Correggi, N. Rougerie. On the Ginzburg-Landau Functional in the Surface Superconductivity Regime. {Commun. Math. Phys.} \textbf{332}:1297--1343 (2014); erratum {Commun. Math. Phys.} {\bf 338}:1451--1452 (2015).
%
\bibitem[CR2]{CR2} M. Correggi, N. Rougerie. Boundary Behavior of the Ginzburg-Landau Order Parameter in the Surface Superconductivity Regime. {\it Arch. Rational Mech. Anal.} {\bf 219}:553--606 (2015).
%
\bibitem[CRY]{1} M. Correggi, N. Rougerie, J. Yngvason. The transition to a giant vortex phase in a fast rotating Bose-Einstein condensate. {\it Commun. Math. Phys.} {\bf 303}:451--508 (2011).
%
\bibitem[dG]{DeG}  P.G. de Gennes. {\it Superconductivity of Metals and Alloys}. Taylor \& Francis (1999).
%
\bibitem[DHM1]{DHM1}	A. Deuchert, C. Hainzl, M.O. Maier. Microscopic derivation of Ginzburg–Landau theory and the BCS critical temperature shift in general external fields. {\it Calc. Var.} {\bf 62}: 203 (2023).
%
\bibitem[DHM2]{DHM2}	A. Deuchert, C. Hainzl, M.O. Maier. Microscopic derivation of Ginzburg-Landau theory and the BCS critical temperature shift in a weak homogeneous magnetic field. {\it Prob. Math. Phys.} {\bf 4}:1--89 (2023).
%
\bibitem[D]{Du} M. Dutour. Phase diagram for Abrikosov lattice. {\it J. Math. Phys.} {\bf 42}: 4915--4926 (2001).
%
\bibitem[FH]{FH} S. Fournais, B. Helffer. 
{\it Spectral methods in surface superconductivity}.
Progr. Nonlinear Differential Equations Appl. {\bf 77},
Birkh\"auser Boston  (2010).
%
\bibitem[FK1]{FK} S. Fournais, A. Kachmar.
The ground state energy of the three dimensional Ginzburg-Landau functional, Part I: Bulk regime.
{\it Commun. Partial Differ. Equ.} {\bf 38}:339--383 (2013).
%
\bibitem[FK2]{FK2} {S. Fournais, A. Kachmar. Nucleation of bulk superconductivity close to critical magnetic field.
{\it Adv. Math.} {\bf 226}:1213--1258 (2011).}
%
\bibitem[FHSS]{3} R.L. Frank, C. Hainzl, R. Seiringer, J.P. Solovej. Microscopic Derivation of Ginzburg-Landau Theory. {\it J. Amer. Math. Soc.} {\bf 25}:667--713 (2012).
%
\bibitem[J]{J} 	R.L. Jerrard. Lower bounds for generalized Ginzburg–Landau functionals. {\it SIAM J. Math. Anal.} {\bf 30}:721--746 (1999).
%
\bibitem[HK]{HK} B. Helffer, A. Kachmar. 
The density of superconductivity in the bulk regime.
{\it Indiana Univ. Math. J.} {\bf 67}:2181--2198 (2018).
%
\bibitem[K1]{K-sima} {A. Kachmar.
The Ginzburg-Landau order parameter near the second critical field.
{\it SIAM J. Math. Anal.} {\bf 46}:572--587 (2014).}
%
\bibitem[K2]{K-jfa} A. Kachmar. 
The ground state energy of the three-dimensional Ginzburg-Landau model in the mixed phase.
{\it J. Funct. Anal.} {\bf 261}:3328--3344 (2011).
%
\bibitem[NR]{4} D.-T. Nguyen, N. Rougerie. Thomas-Fermi profile of a fast rotating Bose-Einstein condensate. {\it Pure Appl. Anal.} {\bf 4}:535--569 (2022).
%
\bibitem[O]{Od} F. Odeh. Existence and Bifurcation Theorems for the Ginzburg-Landau Equations. {\it J. Math. Phys.} {\bf 8}:2351--2356 (1967).
%
\bibitem[P]{Pan} X.-B. Pan. Surface Superconductivity in Applied Magnetic Fields above $\Hcc$, {\it Commun. Math. Phys.} \textbf{228}:327--370 (2002).
%
\bibitem[P]{5} D. P\'erice. Multiple Landau level filling for a mean field limit of 2D fermions. {\it J. Math. Phys.} {\bf 65}:021902, 46 p. (2024).
%
%
\bibitem[RSS]{RSS} {C. Rom\'an, E. Sandier, S.  Serfaty.
Bounded vorticity for the 3D Ginzburg-Landau model and an isoflux problem.
{\it Proc. Lond. Math. Soc.} {\bf 126}:1015--1062 (2023).}
%
\bibitem[R]{6} N. Rougerie. Vortex rings in fast rotating Bose-Einstein condensates. {\it Arch. Rational Mech. Anal.} {\bf 203}:69--135 (2012).
%
\bibitem[Sa]{Sa} E. Sandier. Lower bounds for the energy of unit vector fields and
applications. {\it J. Funct. Anal.} {\bf 152}:379--403 (1998); erratum. {\it J. Funct. Anal.} {\bf 171}:233 (2000).
%
\bibitem[Se]{S99} S. Serfaty. Local minimizers for the Ginzburg–Landau energy near critical magnetic field
part I. {\it Commun. Contemp. Math.} {\bf 1}:213--254 (1999).
%
\bibitem[SS1]{SS00} E. Sandier, S. Serfaty. Global minimizers for the Ginzburg–Landau functional below
the first critical magnetic field. {\it Ann. Inst. H. Poincar\'{e} C -- Anal. Non Lin\'{e}aire} {\bf 17}:119--145 (2000).
%
\bibitem[SS2]{SS001} E. Sandier and S. Serfaty. On the energy of type-II superconductors in the mixed phase. {\it Rev. Math. Phys.} {\bf 12}:1219--1257  (2000).
%
\bibitem[SS3]{SS02}  E. Sandier, S. Serfaty.
The decrease of bulk-superconductivity close to the second critical field in the Ginzburg-Landau model.
{\it SIAM J. Math. Anal.} {\bf 34}:939--956 (2003).
%
\bibitem[SS4]{SS03} E. Sandier, S. Serfaty. Ginzburg-Landau minimizers near the first critical field have bounded vorticity. {\it Cal. Var. PDEs} {\bf 17}:17--28 (2003).
%
\bibitem[SS5]{SS07} E. Sandier, S. Serfaty.
{\it Vortices in the magnetic Ginzburg-Landau model}.
Progr. Nonlinear Differential Equations Appl. {\bf 70}, Birkh\"auser Boston (2007).
%
\hl{\bibitem[ST1]{ST1} I. M. Sigal, T. Tzanateas. Abrikosov lattices at weak magnetic fields, {\it J. Funct. Anal.} {\bf 263}:675--702 (2012).
%
\bibitem[ST2]{ST2} I. M. Sigal, T. Tzanateas. Stability of Abrikosov lattices under gauge-periodic perturbations. {\it Nonlinearity} {\bf 25}:1187--1210 (2012).} 
%
\bibitem[T]{Ti} M. Tinkham. {\it Introduction to Superconductivity}. Dover Publications (2004).

\end{thebibliography}
\end{document}